\RequirePackage[l2tabu, orthodox]{nag}
\documentclass[11pt]{article}

\usepackage{amsmath, mathrsfs, amssymb}

\usepackage[utf8]{inputenc}

\usepackage[a4paper]{geometry}
\usepackage[colorlinks=false, pdfborder={0 0 0}]{hyperref}
\usepackage{graphicx}
\usepackage[usenames,xcdraw,dvipsnames,svgnames,table]{xcolor}
\usepackage{adjustbox}
\usepackage{pifont}
\usepackage{cleveref}
\usepackage{microtype}
\usepackage{siunitx}
\usepackage{booktabs}
\usepackage{enumitem}
\setlist[description]{leftmargin=\parindent,labelindent=\parindent}
\usepackage{subfig}
\usepackage[linewidth=1pt]{mdframed}
\usepackage{multicol, array}
\usepackage{algorithm}
\usepackage{algorithmic}
\usepackage{authblk}

\title{It wasn't me! Plausible Deniability in Web Search}
\author[1]{P\'{O}L MAC AONGHUSA}
\author[2]{DOUGLAS J. LEITH}
\affil[1]{IBM Research \\ Trinity College Dublin}
\affil[2]{Trinity College Dublin}
\date{}                     
\newcommand*\rot{\rotatebox{90}}

\usepackage{array}
\newcolumntype{L}[1]{>{\raggedright\let\newline\\\arraybackslash\hspace{0pt}}m{#1}}
\newcolumntype{C}[1]{>{\centering\let\newline\\\arraybackslash\hspace{0pt}}m{#1}}
\newcolumntype{R}[1]{>{\raggedleft\let\newline\\\arraybackslash\hspace{0pt}}m{#1}}

\newcommand{\SET}[1]{\ensuremath{\mathcal{#1}}}               
\newcommand{\EV}[1]{\ensuremath{\mathscr{#1}}}                
\newcommand{\VEC}[1]{\ensuremath{\mathbf{\bar{#1}}}}                

\def\EPS{\textit{\ensuremath{\epsilon}-Indistinguishability }}

\def\PRI{\textbf{PRI }}
\def\FPRI{\ensuremath{\mathbb{M}_{k}}}
\def\FPRIBAR{\ensuremath{\bar{\mathbb{M}}_{k}}}
\def\FPRIMAX{\ensuremath{\mathbb{M}}_{k}^{*}}
\def\FPRIMIN{\ensuremath{\mathbb{M}}_{k,*}}
\def\FPRIHAT{\ensuremath{\hat{\mathbb{M}}_{k}}}
\def\FPRIHATMIN{\ensuremath{\hat{\mathbb{M}}_{k,*}}}
\def\PRIEPS{\textbf{PRI+} }

\def\P{\text{\rm Prob}}

\newcommand*{\GREEN}[1]{{\color[rgb]{0,1,0.5}\textbf{#1}}}
\newcommand*{\RED}[1]{{\color[rgb]{1,0,0}\textbf{#1}}}
\newcommand*{\BLUE}[1]{{\color[rgb]{0,0,1}\textbf{#1}}}

\newtheorem{theorem}{Theorem}[section]

\newtheorem{proposition}[theorem]{Proposition}

\newenvironment{proof}[1][Proof]{\begin{trivlist}
\item[\hskip \labelsep {\bfseries #1}]}{\end{trivlist}}
\newenvironment{definition}[1][Definition]{\begin{trivlist}
\item[\hskip \labelsep {\bfseries #1}]}{\end{trivlist}}

\newenvironment{assumption}[1][Assumption]{\begin{trivlist}
\item[\hskip \labelsep {\bfseries #1}]}{\end{trivlist}}

\newcommand{\QED}{\nobreak \ifvmode \relax \else
      \ifdim\lastskip<1.5em \hskip-\lastskip
      \hskip1.5em plus0em minus0.5em \fi \nobreak
      \vrule height0.75em width0.5em depth0.25em\fi}

\begin{document}
\maketitle

\begin{abstract}
Our ability to control the flow of sensitive personal information to online systems is key to trust in personal privacy on the internet. We ask how to detect, assess and defend user privacy in the face of search engine personalisation? We develop practical and scalable tools allowing a user to detect, assess and defend against threats to plausible deniability. We show that threats to plausible deniability of interest are readily detectable for all topics tested in an extensive testing program. We show this remains the case when attempting to disrupt search engine learning through noise query injection and click obfuscation are used. We use our model we design a defence technique exploiting uninteresting, proxy topics and show that it provides amore effective defence of plausible deniability in our experiments.
\end{abstract}

\section{Introduction}
\label{sec:intro}

Individual privacy is not just about monitoring the flow of personal information, but also having sufficient agency or control to direct the flow of personal information appropriately, \cite{Nissenbaum:2009:PCT:1822585,boyd2011talk}. When presented with discreditable content, depending on social context, we may wish to deny our interest. We ask when can a user can \emph{plausibly deny} interest in topics on the \emph{balance of probabilities} or with \emph{reasonable doubt}? 

In this paper we extend the work begun in \cite{mac2015don} by improving both the accuracy and robustness of detection techniques. We also show how the same approach can be extended to assess the level of threat in a plausible deniability model. We use our results to design and assess counter-measures to defend plausible deniability against threats.

Our methods are chosen to be straightforward to implement using readily available technologies. Our results are verified in a realistic setting through extensive experiments using Google Search as a source of data. By applying readily implementable methods, we are able to correctly detect evidence of sensitive topic learning of potentially sensitive topics, such as health, finance and sexual orientation, 98\% of the time with noise from false detection rates of below 1\% on average.

In the case of the ability of a user to plausibly deny interest in a topic, we derive guarantees on the best-possible level of plausible deniability a user can expect during web search. We design a range of potential defences of plausible deniability and through an extensive test regime we find that a defence employing uninteresting ``proxy topics'' is effective in protecting plausible deniability in the case of $100\%$ of sensitive topics tested. By contrast, experimental measurements indicate that, by observing as few as 3-5 revealing queries, a search engine can infer a user is interested in a sensitive topic on the balance of probabilities in $100\%$ of topics tested when no effective defence is provided. We show that defence strategies based on random query injection of random noise queries and misleading click patterns may provide some protection for individual, isolated queries, but that search engines are able to learn quickly so seem to offer little or no improvement when considering plausible deniability over the longer term.

The contributions in this paper are:
\begin{itemize}
	\item Technical refinements to our detection technology, called \PRI in \cite{mac2015don}, improving detection accuracy, robustness to changes in training and test data, and simplicity of application.
	\item Guarantees on the applicability of our privacy approach as an aid to design of privacy protection mechanisms. 	
	\item A comprehensive experimental campaign using readily accessible sources to assess the effectiveness of privacy enhancement strategies for individual privacy protection.
\end{itemize}

Our results suggest that search engine capability is continuously evolving so that we can reasonably expect search engines to respond to privacy defences with more sophisticated learning strategies.  Our results also point towards the fact that the text in search queries plays a key role in search engine learning. While perhaps obvious, this observation reinforces the user's need to be circumspect about the queries that they ask if they want to avoid search engine learning of their interests.

\subsection{Related Work}
\label{sec:related:work}
   
The importance of control over appropriate flow of information is discussed extensively in legal and social science fields. Individual control over personal information flow is discussed in a critique of the \emph{nothing to hide} defence for widespread surveillance in \cite{Solove:2007}. Individual privacy and its social consequences are are discussed in \cite{Nissenbaum:2009:PCT:1822585,boyd2011talk}, where agency or control over appropriate disclosure is identified as a concern.

The potential consequences of online profiling and personalisation in censorship and discrimination have received growing attention in the research literature. Personalisation as a form of censorship --  termed a filter bubble in \cite{Pariser:2011:FBI:2029079} -- is explored in \cite{Hannak:2013:MPW:2488388.2488435}. In a filter bubble, a user cannot access subsets of information because the recommender system algorithm has decided it is irrelevant for that user. In \cite{Hannak:2013:MPW:2488388.2488435} a filter bubble effect was detected in the case of Google Web Search in a test with 200 users. 

Discrimination associated with personalisation is investigated in \cite{Sweeney:2013:DOA:2460276.2460278}, where an extensive review of adverts from Google and \url{Reuters.com} showed a strong correlation between adverts suggestive of an arrest record and an individual's ethnicity. Searches containing first names considered black-identifying were on average $25\%$ more likely to receive adverts indicative of an arrest record than searches including white-identifying first names. In \cite{Guha:2010:CMO:1879141.1879152}, the authors considered online advertising targeted {exclusively} to gay men.  Several authors have examined privacy in the context of recommender systems, broadly addressing the question of data privacy with respect to the user data once collected by the recommender. For example in \cite{rama01pr} the authors identify threats from data linking or combination by identifying similar patterns of preference or behaviour conjunction with other data sources to uncover identities and reveal personal details. In their concluding remarks, the authors state that ``the ideal deterrents are better awareness of the issues and more openness in how systems operate in the marketplace. In particular, individual sites should clearly state the policies and methodologies they employ with recommender systems''.

Help for users to better understand the level and impact of personalisation taking place, and to exert some control control over their online privacy has resulted in several technologies. Examples of unsubstantiated and misleading claims by providers of technology to enhance individual privacy are common, \cite{zdnet:charlatans,zdnet:fake}.   Concerns about objective evaluation of the claims by providers of such technologies have attracted the attention of Government, where the need for \emph{``Awareness and education of the users \ldots ''} is identified in \cite{enisa:guide} as a key step to building trust and acceptance of privacy technologies for individuals. Accountability, and enforcement of accountability, for privacy policy is also attracting attention. Regulatory requirements for data handling in industries such as Healthcare (HIPPA) and Finance (GLBA) are well established. The position with respect to handling of data collected by online recommender systems is less clear. In \cite{datta2014privacy}, the author reviews computational approaches to specification and enforcement of privacy policies at large scale.

The research literature on personalisation of web search is extensive, see \cite{decade:web:mining} for an historical survey of results in web mining for personalisation.  In addition to personalisation of adverts based on search queries and pages visited extensive use is already made of data such as location, IP address and browser agent.  Recommender systems also continue to evolve more sophisticated methods of content selection.  For example, semantic targeting techniques, where the overall theme of a web-page is used to select contextually related adverts, are used by companies such as Google Knowledge Graph, \cite{google:graph}, and iSense, \cite{isense}, while Zemanta, \cite{zemanta}, provides a browser add-on that suggests semantically relevant articles, images, links and keywords to content creators.   As a result even so-called `private' browsing mode may not prevent personalisation; for example in \cite{Aggarwal:2010:APB:1929820.1929828} the authors investigate how a range of popular browser extensions and plugins undermine the security of private browsing.

Several approaches exist for obfuscating user interactions with search engines with the aim of distrupting online profiling and personalisation. However this work is at an early stage.  GooPIR, \cite{Domingo-Ferrer:2009h,SaNchez:2013:KSC:2383079.2383150}, attempts to disguise a user's ``true''  queries by adding masking keywords directly into a true query before submitting to a recommender system. Results are then filtered to extract items that are relevant to the user's original true query.  PWS, \cite{balsa2012ob}, and TrackMeNot, \cite{howe2009trackmenot,peddinti2011limitations}, inject distinct noise queries into the stream of true user queries during a user query session, seeking to achieve an acceptable level of anonymity while not overly upsetting overall utility. Browser add-ons to aid user awareness, such as Mozilla Lightbeam, \cite{lightbeam}, and PrivacyBadger, \cite{badger}, help a user understand where their data is shared with third parties through the sites they visit. XRay, \cite{DBLP:journals/corr/LecuyerDLPPSCG14}, reports high accuracy in identifying which sources of user data such as email or web search history might have {triggered} particular results from online services such as adverts.  Search engine algorithm evolution regarded as a continuous ``arms-race'', is evidenced in the case of Google, for example, by major algorithm changes such as \emph{Caffeine} and \emph{Search+ Your World} have included additional sources of background knowledge from Social Media, improved filtering of content such as \emph{Panda} to counter spam and content manipulation, most recently semantic search capability has been added through \emph{Knowledge Graph} and \emph{HummingBird}, \cite{search_timeline_1}, \cite{search_history}, \cite{search_history_2}.

Evaluation of the effectiveness of privacy defence in the wild was performed by \cite{peddinti2010privacy} in the case of \emph{TrackMeNot} where the authors discovered that demonstrate that by using only a short-term history of search queries it is possible to break the privacy guarantees of TrackMeNot using readily available machine-learning classifiers. The importance of background information in user profiling is explored in \cite{Petit2016}. Here a similarity metric measuring distance between \emph{known} background information about a user, given by query history, and subsequent queries is shown to identify 45.3\% of TrackMeNot and 51.6\% of GooPIR queries. Self-regulation has also proven problematic, in \cite{balebako2012measuring}, six different privacy tools, intended to limit advertising due to behavioural profiling, are assessed. The tools assessed implement a variety of tactics including cookie blocking, site blacklisting and Do-Not-Track (DNT) headers. DNT headers were found to be ineffective in tests at protecting against adverts based on user profiling. 

The importance of user click data for content ranking and placement has been extensively studied, see for example \cite{grotov2015comparative} for a recent review.  Research on privacy protection of user click actions has proposed targeted advertising systems where users can send perturbed click feedback through an intermediate agent or broker \cite{wang2015privacy}, seeking to ensure acceptable user privacy while protecting the revenue models of the back-end recommender systems which are based on accurate click tracking. In \cite{degeling2016your} the authors show that an approach called \emph{anti-profiles} whereby uninteresting ``dummy'' topics are constructed based on user interests which can be used to generate noise queries disrupting search engine learning. The authors show that selecting noise queries based on uninteresting topics provides better privacy protection than obfuscation by injecting randomly selected noise queries.

We model a search engine as a black-box by making minimal assumptions about its internal workings. The technique of using predefined \emph{probe queries}, injected at intervals into a stream of true user queries as fixed sampling points, was used in \cite{mac2015don}. The content provided in response to probe queries can be tested for signs of recommender system learning. The technique of using predefined probe queries  is borrowed from black-box testing. Modelling an adversary as a black-box, where internal details of recommender systems algorithms and settings are unknown to users, is mentioned in several sources \cite{datta2014privacy} and \cite{Hannak:2013:MPW:2488388.2488435}.  

\section{General Setup}
\label{sec:gen:setup}
\subsection{Overview}
\label{sec:notation} 

We wish to detect threats to privacy due to personalisation, assess the degree of threat posed to plausible deniability when personalisation is detected, and, provide some ability for the user to defend against unwanted personalisation. We structure our approach as follows:

\begin{description}
	\item [Detection Model] The detection model considered in this paper is one of topic category \emph{distinguishability} where a search engine  seeks to determine a user's likely interest in commercially valuable categories -- rather than individually identify a single user.

	\item [Assessment Model] We assess the level of severity of detected threats by testing the ability of a user to gauge the degree of threat to plausible deniability, that is, ability to assert or deny interest in topics deemed sensitive.

	\item [Defence Models] We seek practical defences against threats to plausible deniability by providing a user with agency to control  personalisation when it is assessed as unwanted. 
 
\end{description}

Our initial assumption is that a for-profit commercial recommender system, such as a search engine, selects variable page content to maximise its expected revenue.  In particular, when a search engine infers that a particular advertising category is likely to be of interest to a user, and it is more likely to generate click through and sales, it is obliged to use this information when selecting which adverts to display. This suggests that, by examining \emph{advert content} recommended by the search engine, it is possible to detect evidence of sensitive topic profiling by the search engine.

To test for learning a predefined \emph{probe query} is injected into a stream of ``true'' queries during a query session. In this way, any differences detected in advert content in response to probe queries can be compared to identify evidence of learning.  

\subsection{Notation}
Let $\{c_1, \ldots, c_N\}$ denote a set of sensitive categories of interest to an individual user, \emph{e.g.} bankruptcy, cancer, addiction, \emph{etc}. Gather all other uninteresting categories into a catch-all category denoted $c_{0}$. The set $\SET{C} = \{c_0, c_1, \ldots, c_N\}$ is \emph{complete} in the sense that all user topic interests can be represented as subsets of $\SET{C}$ with the usual set operations.

We assume that a user interacts with a search engine by issuing a query, receiving a web page in response and then clicking on one or more items in the response. A single such interaction, labeled with index $j$, consists of a \emph{query, response page, item-click} triple, denoted 
	$\Omega_j = \left(q_j, p_j, l_j \right)$.
A user session of length $k>0$ steps consists of a {sequence} of $k$ individual steps, and is denoted  $\left\{\Omega_k\right\}_{k\ge1}$.   
The sequence of interactions $\left\{\Omega_k\right\}_{k\ge1}$ is jointly observed by  the user and the search engine -- and perhaps several other third-party observers. 

Let $\EV{E}_k$ denote the prior evidence -- also referred to as background knowledge -- available to an observer at the start of step $k$.  We assume the search engine does not change its background knowledge during a user session other than through $\Omega_k$.   That is, $\EV{E}_1$ denotes the prior evidence available to the search engine before the user session begins, \emph{e.g.} the user's login profile, historical queries, weblogs  \emph{etc}, and for $k=1,2,\cdots$ we have:
\begin{align}
\EV{E}_{k+1} &= \{\EV{E}_{k}, \Omega_{k} \}_{k\geq1} \notag
\end{align}%

At each step $k$ of a query session, a search engine can construct a prior probability for evidence of user interest in categories in $\SET{C}$, represented by the interest vector $\VEC{x} \in \{0,1\}^{N}$, namely: 
\begin{align}
\P(\VEC{X} = \VEC{x} \, \vert  \,  \EV{E}_k)
\label{eqn:actions:001}
\end{align}

\noindent Where $\VEC{X}$ is a random variable with sample space $\{0,1\}^{N}$ and a value of $1$ in element $i$ of \VEC{X} indicates evidence is detected of user interest in topic $c_i$, and $\EV{E}_k$ denotes the prior evidence available to an observer at the start of step $k$ in a query session. We have assumed that user interests do not vary over the lifetime of a single query session so that $X$ is independent of the step $k$. 

The individual interest vector for topic $i$ is denoted $\VEC{c}_i$, a vector with a single $1$ in the $i^{th}$ position and $0$ in all other positions. The prior probability of detecting evidence that a user is interested \emph{only} in topic $i$ at the beginning of the session is therefore:
\begin{align*}
	\P(\VEC{X} = \VEC{c}_i \vert \, \EV{E}_1) := p_i
\end{align*}

On observing events $\Omega_k$ from the user at step $k$, a search engine can construct a posterior probability distribution for $X$ conditioned on observing $\Omega_k$. Denote the posterior probability of detecting evidence that a user is interested only in topic $i$ at query step $k$, conditioned on observing $\Omega_k$ and background knowledge $\EV{E}_k$, by:
\begin{align*}
	\P(\VEC{X} = \VEC{c}_i \vert \, \Omega_k, \, \EV{E}_{k} ) := \pi_{i,k}
\end{align*}
Since $\mathcal{C}$ contains all topics: 
\begin{align}
	\sum_{i=0}^{N} p_i &=1, \qquad \sum_{i=0}^{N} \pi_{i,k} =1 \label{eqn:pri:2}
\end{align}
We adopt the following notation for the maximum and minimum prior and posterior probabilities:
\begin{align*}	
p^{*} &:= \max_{0\leq i \leq N} \{p_i\}, \quad p_{*} := \min_{0\leq i \leq N} \{p_i\},\quad \pi^{*}_{k} := \max_{0\leq i \leq N} \{\pi_{i,k}\}, \quad \pi_{k,*} := \min_{0\leq i \leq N} \{\pi_{i,k}\}
\end{align*}

\subsection{Detection Model}
\label{sec:priv:model}

The following \emph{indistinguishability} definition of \emph{privacy risk} measures the change in belief by a search engine due to inference from observed user events relative to a prior probability  conditioned on the background data available at the start of the query session:

\begin{definition}[\EPS]
\label{def:eps:001}
A user is said to be $\epsilon$-Indistinguishable, given an interest vector $\VEC{x} \in \{0,1\}^{N}$, with respect to an observation of user actions $\Omega_k$, and for $\epsilon > 0$ if:

\begin{align}
e^{-\epsilon} &\leq \FPRI (\VEC{x}, \Omega_k) \leq e^{\epsilon}
\label{def:eps:indist:001}
\end{align}

\noindent where

\begin{align}
\FPRI (\VEC{x}, \Omega_k)  &= \frac{\P(\VEC{X} = \VEC{x} \, \vert  \,  \Omega_k, \EV{E}_{k} )}{\P(\VEC{X} = \VEC{x} \, \vert  \,  \EV{E}_{1})}
\label{eqn:m:u}
\end{align}
\noindent is called the \emph{\EPS Score} of the interest vector $\VEC{x}$ given observation $\Omega_k$ and $\EV{E}_{k}$ represents the background knowledge available at step $k$.
\label{def:eps:indist1}
\end{definition}
In other words, for \EPS to hold t step $k$ of a query session, for a given combination of categories in $\SET{C}$, the conditional posterior distribution should be approximately equal to the prior distribution at the beginning of the query session for that combination of categories. It is possible that the denominator of \eqref{eqn:m:u} is zero while the numerator is non-zero, and vice versa, so that $0 \leq \FPRI (\VEC{x}, \Omega_k) < \infty $.

The corresponding \EPS score for topic $c_i$ at step $k$ is therefore:
\begin{align}
	\FPRI(\VEC{c}_i, \Omega_k) &=\frac{\pi_{i,k}}{p_i}
	\label{eqn:pri:3}
\end{align}
\noindent where it is assumed there is some user interest in each topic $c_i \in \SET{C}$ so that $p_i > 0$ for all $i=0,1, \ldots N$ and \eqref{eqn:pri:3} is well-defined for each topic $c_i$.

Let:
\begin{align*}
	\FPRIBAR &= \frac{1}{N+1}\sum_{i=0}^{N}\FPRI(\VEC{c}_i, \Omega_k), \quad \FPRIMAX = \max_{0\leq i \leq N} \{\FPRI(\VEC{c}_i, \Omega_k)\}, \quad\FPRIMIN = \min_{0\leq i \leq N} \{\FPRI(\VEC{c}_i, \Omega_k)\}
\end{align*}
denote the average, maximum and minimum values of the \PRI scores, $\FPRI$, at step $k$ in the query session. The observed quantity, $\Omega_k$ is omitted in the expressions for $\FPRIBAR$, $\FPRIMAX$ and $\FPRIMIN$ for notational simplicity since its value is fixed once the  query step $k$ is set. 

The following relationships between these quantities allow us to relate prior and posterior probabilities through the quantities $\FPRIBAR$, $\FPRIMAX$ and $\FPRIMIN$:
\begin{proposition}
\label{prop:0}
\begin{align}
	\frac{\pi_{k,*}}{p^{*}} &\underset{(a)}{\leq} \FPRIMIN \underset{(b)}{\leq} \FPRIBAR \underset{(c)}{\leq} \FPRIMAX \underset{(d)}{\leq} \frac{\pi^{*}_{k}}{p_{*}} \label{eqn:big:inequality}
\end{align}	
if \EPS also holds then:

\begin{align}
p_{*}e^{-\epsilon} &\leq \pi_{k,*} \leq p^{*}\FPRIMIN, \quad p_{*}\FPRIMAX \leq \pi^{*}_{k} \leq p^{*}e^{\epsilon}
\label{eqn:little:inequality} 	
\end{align}
in the special case, when all $N+1$ categories are a priori equally likely:
	\begin{align}
		\pi_{k,*} &= \frac{\FPRIMIN}{N+1}, \quad
		\pi^{*}_{k} = \frac{\FPRIMAX}{N+1} \label{eqn:uniform:prior}
	\end{align}
	
\end{proposition}
\begin{proof}
Inequalities (b) and (c) follow directly from the definitions of $\FPRIBAR$, $\FPRIMAX$ and $\FPRIMIN$. Inequalities (a) and (d) follow since $\frac{\pi_{k,*}}{p^{*}} \leq \FPRI(\VEC{c}_i, \Omega_k)$ and  $\frac{\pi^{*}_{k}}{p_{*}} \geq \FPRI(\VEC{c}_i, \Omega_k)$
for all $i = 0,1, \ldots N$ and given $\Omega_k$. Consequently $\frac{\pi_{k,*}}{p^{*}} \leq \FPRIMIN$ and $\frac{\pi^{*}_{k}}{p_{*}} \geq \FPRIMAX$.

To show \eqref{eqn:little:inequality}, rearrange inequalities (a) and (d) in \eqref{eqn:big:inequality} and since \eqref{def:eps:indist:001} implies 
\begin{align*}
p_{*}e^{-\epsilon} &\leq p_{i} e^{-\epsilon} \leq \pi_{i} \leq p_{i} e^{\epsilon} \leq p^{*} e^{\epsilon}, \quad \forall i=0, 1, \ldots N \\
&\implies p_{*}e^{-\epsilon} \leq \pi_{k,*} \leq  \pi^{*}_{k} \leq p^{*} e^{\epsilon}
\end{align*}

Now \eqref{eqn:uniform:prior} follows from the definition of $\FPRI$ in \eqref{eqn:pri:3} by noting that $p_{*} = p^{*} = \frac{1}{N+1}$ when $p_i = \frac{1}{N+1}$ so that $\pi_{k,*} = \FPRIMIN/(N+1)$ and $\pi^{*}_{k} = \FPRIMAX/(N+1)$.
\QED 	
\end{proof}%

The total variation of \FPRI is bounded by the prior probabilities $p_i$ as the following result shows.
\begin{proposition}
\label{prop:1} 
\begin{align}
\frac{1}{p^{*}} &\leq \sum_{i=0}^{N} \FPRI(\VEC{c}_i, \Omega_k) \leq \frac{1}{p_{*}}
\label{eqn:bounds:1}	
\end{align}

If all $N+1$ categories are a priori equally likely, so that $p_i = \frac{1}{N+1}$ for $i=0,1,\dots,n$, then:
	\begin{align}
	\FPRIBAR = 1 \label{eqn:step:3:1}
	\end{align}
\end{proposition}%
\begin{proof}
By direct substitution of \eqref{eqn:pri:3} into \eqref{eqn:bounds:1}, $\sum_{i=0}^{N} \FPRI(\VEC{c}_i, \Omega_k) = \sum_{i=0}^{N} \frac{\pi_{i,k}}{p_i}$. From the definitions of $p^{*}$ and $p_{*}$:
\begin{align}
	\frac{1}{p^{*}}\sum_{i=0}^{N} \pi_{i,k} &\leq \sum_{i=0}^{N} \frac{\pi_{i,k}}{p_i} \leq \frac{1}{p_{*}}\sum_{i=0}^{N} \pi_{i,k} 
	\label{eqn:bounds:step:2}
\end{align}
Next, since $\sum_{i=0}^{N} \pi_{i,k} = 1$ then $\frac{1}{p^{*}} \leq \sum_{i=0}^{N} \frac{\pi_{i,k}}{p_i} \leq \frac{1}{p_{*}}$ and \eqref{eqn:bounds:1} follows.

To prove \eqref{eqn:step:3:1}, since $p_i = \frac{1}{N+1}$, $p_{*} = p^{*} = \frac{1}{N+1}$ then from \eqref{eqn:bounds:1} we have that $\sum_{i=0}^{N} \FPRI(\VEC{c}_i, \Omega_k) = N+1$ and \eqref{eqn:step:3:1} follows directly.  
\QED 
\end{proof}%


\subsection{Assessment Model}
\label{sec:by:design}

We adopt a threat assessment model based on \emph{plausible deniability} to capture a notion of reasonable doubt. That is, with the usual setup of topics $\SET{C} = \{c_0, c_1, \ldots c_N\}$, we require the probability that the user is interested in any one of the topics in $\SET{C}$ to be close to the probability that the user is interested in any other topic in $\SET{C}$ as follows:

\begin{definition}[Mutual Plausible Deniability\label{def:plausible}]
\label{def:plausible:detect:001}
Given a set of topics $\SET{C} = \{c_0, c_1, \ldots c_N\}$, 
a user is said to have \emph{mutual plausible deniability} in topics in \SET{C} with confidence  $\delta\geq 0$, with respect to an observation of user actions $\Omega_k$  at step $k$ in the query session, if:

\begin{align}
\max_{i \neq j}\left\vert \P(\VEC{X} = \VEC{c}_i \, \vert  \,  \Omega_k, \EV{E}_{k} ) - \P(\VEC{X} = \VEC{c}_j \, \vert  \,  \Omega_k, \EV{E}_{k} ) \right\vert &\leq \delta_{k}
\label{eqn:delta:deny}
\end{align}
\label{def:delta:deny}
\noindent where $\delta_k$ is a constant, independent of the choice of topic $c_i \in \SET{C}$.
\end{definition}

Adopting the usual notation of $\pi_{i,k} := \P(\VEC{X} = \VEC{c}_i \, \vert  \,  \Omega_k, \EV{E}_{k} )$, \eqref{eqn:delta:deny} reads:
\begin{align}
	\max_{i\neq j} \left\vert \pi_{i,k} - \pi_{j,k} \right\vert < \delta_{k}	
	\label{eqn:delta:deny:2}
\end{align}

The following result confirms the intuitive interpretation that categories that are mutually plausibly deniable should have approximately equal posterior probability after an observation is made. In other words, $\delta_k$ measures how closely interest in topic category $c_k$ approaches random chance.

\begin{proposition}\label{prop:mpd:1}
Given a set of topics $\SET{C} = \{c_0, c_1, \ldots c_N\}$ satisfying Mutual Plausible Deniability for $\delta_k > 0$, then:
\begin{align}
\left\vert \pi_{i,k} - \frac{1}{N+1} \right\vert < \delta_k, \quad \forall i=0,1,\ldots N	
\label{eqn:delta:deny:3}
\end{align}	
\end{proposition}
\begin{proof}
Taking $0 \leq J \leq N$,
	\begin{align}
		\left\vert (N+1)\pi_{J,k} - 1 \right\vert &= \left\vert (N+1)\pi_{J,k} - \sum_{i=0}^{N}\pi_{i,k} \right\vert  &\text{(Since $\sum_{i=0}^{N} \pi_{i,k} = 1$)} \\
		&= \left\vert \sum_{i=0}^{N} \left( \pi_{J,k} - \pi_{i,k} \right)  \right\vert \\
		&\leq \sum_{i=0}^{N}\left\vert \pi_{J,k} - \pi_{i,k} \right\vert  &\text{(Triangle inequality)} \\
		&< (N+1)\delta_k &\text{(From \eqref{eqn:delta:deny})}
	\end{align}
So that
	\begin{align}
	\left\vert \pi_{J,k} - \frac{1}{N+1} \right\vert < \delta_k	
	\label{eqn:delta:deny:3:step:9}
	\end{align}

Since the choice of $J$ is arbitrary, \eqref{eqn:delta:deny:3:step:9} applies for any $0 \leq J \leq N$ and the result follows. \QED 
\end{proof}

The next result provides a method to bound  mutual plausible deniability where there is equal a priori probability for all topics so that $p_i = 1/(N+1)$ .

\begin{proposition}\label{prop:mpd:2}
	Given a set of topics $\SET{C} = \{c_0, c_1, \ldots c_N\}$ with equal a priori probability for all topics, Mutual Plausible Deniability holds for any $\delta_k$ satisfying,
	\begin{align}
		\delta_k &\geq \frac{\FPRIMAX - \FPRIMIN}{N+1}
		\label{eqn:mpd:2}
	\end{align}
	So that ~\eqref{eqn:mpd:2} gives the \emph{best case} level of mutual plausible deniability possible for $\SET{C}$.
\end{proposition}
\begin{proof}
	From \ref{eqn:delta:deny:2}, 
	\begin{align}
		\max_{i \neq j}\left\vert \pi_{i,k} - \pi_{j,k} \right\vert &= \left\vert \pi^{*}_{k} - \pi_{k,*} \right\vert < \delta_k \label{prop:mpd:1:step1}
	\end{align}
Since prior probabilities are equal, it follows from \eqref{eqn:uniform:prior} in Proposition~\ref{prop:0} that

\begin{align}
\left\vert \pi^{*}_{k} - \pi_{k,*} \right\vert &= \pi^{*}_{k} - \pi_{k,*} = \frac{\FPRIMAX - \FPRIMIN}{N+1}
\label{prop:mpd:1:step3}
\end{align}

\noindent and the result follows. \QED
\end{proof}

From \eqref{eqn:mpd:2}, the \emph{best possible} mutual plausible deniability over $K>0$ measurements, denoted $\delta^{*}$, is given by:

\begin{align}
	\delta^{*} &= \max_{0<k\leq K}\frac{\FPRIMAX - \FPRIMIN}{N+1}
	\label{eqn:mpd:3}
\end{align}
We will report values of $\delta^{*}$ using \eqref{eqn:mpd:3} in subsequent experimental results.

%
%
%

\subsection{Defence Models}
\label{sec:defence:model}
We turn now to model defences in the face of threat to plausible deniability when all categories are a priori equally likely. Substituting \eqref{eqn:uniform:prior} in \eqref{eqn:mpd:3} implies that: 
\begin{align}
\delta^{*} &= \max_{0<k\leq K}(\pi^{*}_{k} - \pi_{k,*})
\label{eqn:design:01}
\end{align}
\noindent suggesting that it is beneficial for defence to attempt to distribute the posterior probability mass as evenly as possible over all $N+1$ topic categories to minimise \eqref{eqn:design:01}. As both user clicks and queries are sources of information for the search engine we will consider both in designing candidate defence mechanisms.

We consider the following possible defence strategies against threats to plausible deniability: 

\begin{description}
	\item [a) Random Noise] An obvious approach is to select ``noise queries'' randomly from a list of popular queries and inject them into the ``sensitive'' query session. We reason that injecting such randomly selected queries emulates user interest in the ``other'' catchall topic $c_0$ which may be sufficient to reduce $\delta^{*}$. A key challenge is to identify the correct type and amount of ``noise'' queries to introduce in order to stimulate search engine learning away from the ``true'' sensitive topic. We will therefore test the effect of a number of different ratios of noise-to-sensitive query levels. 
	\item [b) Proxy Topic Noise] Inserting sequences of noise queries, with the queries in each sequence related to a selected uninteresting ``proxy topic'', may stimulate search engine learning in favour of the ``proxy topic''.  Such a change would then be detectable as evidence of interest in the ``proxy topic'' using \PRIEPS. To simulate interest in ``proxy topics'', we must be careful to provide sufficient ``proxy topic'' queries to confirm our interest.
	\item [c) Click Strategies] In \cite{mac2015don}, the authors note that varying click patterns is observed to change the absolute volume of adverts appearing on a page. We reason that, since advert space is limited,  perhaps varying click responses to content will simulate possible interest in a wider range of content on result pages.
\end{description}

The difference between a) and b) is that in a) each noise query is chosen at random from a list of popular queries without reference to any single topic, whereas in b) a block of ``noise queries'' all related to a single proxy topic is generated and inserted into the user session.

\section{Implementation}
\label{sec:experiment:setup}

\subsection{Preliminaries}
\label{sec:bayesian:formulation}

We recap material from \cite{mac2015don} about the construction of the \PRI estimator. The first assumption is that the background knowledge at the first step of a query session, $\EV{E}_1$, provides sufficient description of background knowledge for all subsequent steps of that query session, $\EV{E}_k$.

\begin{assumption}[Sufficiently Informative Response]\label{a:1}
Let $\SET{K}\subset\{1,2,\cdots\}$ label the subsequence of steps at which a probe query is issued.   
At each step $k\in\SET{K}$ at which a probe query is issued,
\begin{align}
\frac{\P(\VEC{X} = \VEC{c} | \Omega_k, \EV{E}_{k} )}{\P(\VEC{X} = \VEC{c} |  \EV{E}_{1} )}&=\frac{\P(\VEC{X} = \VEC{c} | \Omega_k,\EV{E}_{1} )}{\P(\VEC{X} = \VEC{c}|  \EV{E}_{1})}
\label{eqn:assumption:informative}
\end{align}
\end{assumption}

\noindent That is, it is not necessary to explicitly use knowledge of the search history during the current session when estimating $\FPRI$ for a topic $c$ as this is already reflected in the search engine response,  $\Omega_k$, and the background knowledge $\EV{E}_1$, at step $k$.  Assumption \ref{a:1} greatly simplifies estimation as it means we do not have to take account of the full search history, but requires that the response to a query reveals search engine learning of interest in sensitive category $c$ which has occurred. Assumption~\ref{a:1} was called the ``Informative Probe'' assumption in \cite{mac2015don}.

Our next assumption is that adverts are selected to reflect search engine belief in user interests. In this way adverts are assumed to be the principal way in which search engine learning is revealed. Given this assumption, conditional dependence on $\Omega_k$ can be replaced with dependence on the adverts appearing on the screen.

\begin{assumption}[Revealing Adverts]\label{a:2} In the search engine response to a query at step $k$ it is the adverts ${a}_k$ on response page $p_k$ which primarily reveal learning of sensitive categories. 
\begin{align}
\P(\VEC{X} = \VEC{c} | \Omega_k,\EV{E}_{1}) = \P(\VEC{X} = \VEC{c} | {a}_k,\EV{E}_{1} ), \ k\in \SET{K}\label{eq:assum0}
\end{align}
\end{assumption}

We estimate background knowledge $\EV{E}_1$ by selecting a training dataset, denoted $\SET{T}$, consisting of (label,advert) pairs; where the label is the category in $\SET{C}$ associated with the corresponding advert. For example, when testing for evidence of a single, sensitive topic, called ``Sensitive'', $\SET{T}$ contains items labelled ``Sensitive' or ``Other'', where ``Other'' is the label for the uninteresting, catch-all topic $c_0$. In this way $\SET{T}$ approximates the prior observation evidence available at the start of the query session so that $\SET{T}$ is an estimator for $\EV{E}_1$.

Text processing of $\SET{T}$ produces a dictionary $\SET{D}$ of \emph{keyword} features.  This processing removes common English language high-frequency words and maps each of the remaining keywords to a stemmed form by removing standard prefixes and suffixes such as ``--ing'' and ``--ed''.  The dictionary $\SET{D}$ represents an estimate of the known universe of keywords according to the background knowledge contained in the training data.

Text appearing in the adverts in a response page is preprocessed in the same way as \SET{T} to produce a sequence of keywords from $\SET{D}$ for each advert; denoted $W=\{w_1,w_2,\cdots,w_{|W|}\}$. Let $n_{\SET{D}}(w | W) := \vert \{i:i \in \{1,\cdots,|W|\},w_i = w\} \vert$, denote 
 the number of times with which an individual keyword $w \in \SET{D}$ occurs in a sequence $W=\{w_1,w_2,\cdots,w_{|W|}\}$. The relative frequency of an individual keyword $w\in\SET{W}$ is therefore,
\begin{align}
\phi_{\SET{D}}(w | W) &= 
\frac{n_{\SET{D}}(w | W)}{\sum_{w\in\SET{D}} n_{\SET{D}}(w | W)}
\label{eqn:phi:define}
\end{align}
\noindent recalling that only keywords $w$ appearing $\SET{D}$ are admissible due to the text preprocessing.

Let $c_i \in  \SET{C}$ be a sensitive topic of interest,  
and let $\SET{T}(c_i)$ denote the subset of $\SET{T}$ where the labels corresponds to $c_i$. Let $T(\SET{C})$ denote the set of adverts labelled for any topic in $\SET{C}$. The \PRI estimator for $\FPRI(\VEC{x}, \Omega_k)$ given adverts $a_k$ appearing on the result page for query number $k$, is\footnote{Note that in \cite{mac2015don} the expression given for $\hat{M}_k (\VEC{c}_i, \Omega_k)$ is incorrect and is corrected here.}:

\begin{align}
\hat{M}_k (\VEC{c}_i, \Omega_k) &= \sum_{w\in{\SET{D}}}\left( \frac{\phi_{\SET{D}}(w | \SET{T}(c_i))}{  \phi_{\SET{D}}(w | \SET{T}(\SET{C})) } \cdot {\phi_{\SET{D}}(w|a_k)} \right) \label{eqn:mk001}
\end{align}
where we concatenate all of the advert text on page $k$ into a single sequence of keywords and $\psi_{\SET{D}}(w|a_k)$ is the relative frequency of $w$ within this sequence.  Similarly, concatenating all of the keywords in the training set $ \SET{T}(c_i)$, respectively $ \SET{T}(\SET{C})$, into a single sequence then $\phi_{\SET{D}}(w | \SET{T}(c_i))$, respectively $\phi_{\SET{D}}(w | \SET{T}(\SET{C}))$, is the relative frequency of $w$ within that sequence.

\subsection{Tuning the \PRI Estimator}
\label{sec:sparsity}

We begin by noting that the quantity $\psi_{\SET{D}}(w|a_k)$ in expression for the \PRI estimator, \eqref{eqn:mk001}, is problematic when the adverts $a_k$ on page $k$ do not contain any of the topic keywords in dictionary $\SET{D}$ i.e. when $a_k=\emptyset$.  Such adverts indicate that there is no detectable evidence of a particular topic and so, to be consistent with the definition of \EPS in Section \ref{def:eps:001}, should result in a \PRI score of one for that topic.  We therefore replace $\phi_{\SET{D}}(w|a_k)$ with 

 \begin{align}
\psi_{0, \SET{D}}(w | a_k) &= 
\begin{cases}
\phi_{\SET{D}}(w | {a_k})
& \mbox{if } a_k\ne \emptyset \\
1 
& \mbox{otherwise}  
\end{cases} 
\end{align}
We further note that, since the training data is based on a limited sample of adverts, for a particular topic we may be unlucky and during the training phase observe no adverts containing an infrequently occurring keyword.   We therefore adopt the following regularisation approach, defining
\begin{align}
n_{\lambda, \SET{D}}(w | W) &= \lambda + (1-\lambda)n_{\SET{D}}(w | W) \\
\phi_{\SET{\lambda,D}}(w | W) 
&= \frac{n_{\lambda,\SET{D}}(w | W)}{\sum_{w\in\SET{D}} n_{\lambda, \SET{D}}(w | W)}
\\
\psi_{\lambda, \SET{D}}(w | a_k) &= 
\begin{cases}
\phi_{\lambda,\SET{D}}(w | {a})
& \mbox{if } a_k\ne \emptyset \\
1 
& \mbox{otherwise}  
\end{cases} 
\label{eqn:phi:define2}
\end{align}
The parameter $0\le \lambda < 1$ enforces a minimum frequency of $1/|\SET{D}|$ on every keyword.  
The expression \eqref{eqn:mk001} is adjusted correspondingly to give a new estimator we call \PRIEPS:

\begin{align}
\FPRIHAT (\VEC{c_i}, \Omega_k) &= \sum_{w\in{\SET{D}}}\left( \frac{ \phi_{\lambda, \SET{D}}(w | \SET{T}(c_i))}{ \phi_{\lambda, \SET{D}}(w | \SET{T}(\SET{C})) } \cdot  \psi_{\lambda, \SET{D}}(w | {a_k}) \right) \label{eqn:mk001:2}
\end{align}
We will use the \PRIEPS estimator, given by \eqref{eqn:mk001:2}, from now on in this paper, unless stated otherwise.

We select the value of regularisation parameter $\lambda$ as follows.   Recall $\SET{T}(c_i)$ denotes the subset of $\SET{T}(\SET{C})$ that is labelled for category $c_i$ and $ n_{\lambda,\SET{D}}(w | \SET{T}(c_i))$ estimates the number of times that keyword $w\in\SET{D}$ occurs when counting over {all} of the advert text in elements of $\SET{T}(c_i)$.   Similarly, $ n_{\lambda,\SET{D}}(w | \SET{T(\SET{C})}) $ estimates the count in the full training set $\SET{T}(\SET{C})$. It follows that
\begin{align}
	\label{eqn:jpdf:prior:1}
	\hat{f}_{\lambda,i}(w) &:= \frac{n_{\lambda,\SET{D}}(w | T(c_i))}{\sum_{w\in\SET{D}} n_{\lambda, \SET{D}}(w | T(\SET{C}))} 
\end{align}
\noindent is an estimator for the joint prior probability $\P(w\in A_1, \, \VEC{X} = \VEC{c}_i \vert  \, \EV{E}_1)$, where $A_1$ is a random variable consisting of the sequence of keywords from $\SET{D}$ appearing in adverts on page 1. That is, $\P(w\in A_1, \, \VEC{X} = \VEC{c}_i \vert  \, \EV{E}_1)$ is the joint probability that keyword $w$ is displayed on the first page of a query session and the user is interested in category $c_i$ given prior evidence $\EV{E}_1$.  The prior probability $p_i:= \P(\VEC{X} =  \VEC{c}_i \vert  \, \EV{E}_1)=\sum_{w\in\SET{D}} \P(w\in A_1, \, \VEC{X} = \VEC{c}_i \vert  \, \EV{E}_1)$ can be estimated by
\begin{align}
	\label{eqn:prior:jpdf:estimator}	
	\hat{p}_{\lambda,i} := \sum_{w\in\SET{D}} f_{\lambda,i}(w) 
\end{align}
The value of the parameter $\lambda$ is chosen so the prior probability distribution is approximately equiprobable by minimising the square error loss function, i.e. 
\begin{align}
\label{eqn:prior:estimator}
\lambda \in \arg \min _{\lambda\ge 0} L(\lambda):= \sum_{i=0}^{N}\left(\frac{1}{N+1} - \hat{p}_{\lambda,i}\right)^{2}
\end{align}



\begin{figure}[tph]
    \centering
    	\caption{\PRIEPS values, plotted by probe step. \textbf{**are the axes ranges here not a little too large -- makes it hard to see the detail ?}} 
    	\label{fig:base:random:pri}
        \includegraphics[width=\columnwidth, height=0.8\columnwidth]{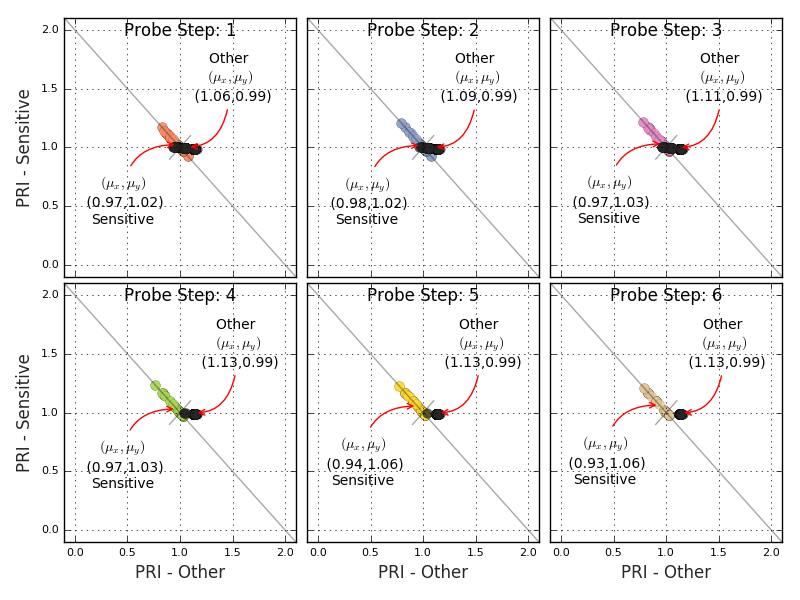}    
\end{figure}%

We illustrate the impact of this regularisation in Figure~\ref{fig:base:random:pri} using a sample of data from our experiments.   The output from \PRIEPS, is an estimate $\FPRIHAT (\VEC{c_i}, \Omega_k)$ for each topic category $i=0,\ldots N$.  We gather these estimates into an $N+1$ dimensional vector $\VEC{P}$ which we call the \PRIEPS \emph{score}, so that the $i^{th}$ component $P_i$ of the \PRIEPS score is $\FPRIHAT (\VEC{c_i}, \Omega_k)$.   In these experiments there is a single sensitive category $c_1$, thus $\SET{C}=\{c_0,c_1\}$ and $\VEC{P}$ is a vector with two elements.   Probe queries are injected at intervals into user sessions, the user sessions themselves consisting either of a sequence of sensitive queries associated with $c_1$ or a sequence of non-sensitive queries associated with $c_0$.   In Figure~\ref{fig:base:random:pri} measured values of $\VEC{P}$ are plotted for the responses to these probe queries, with values from sensitive sessions labeled ``Sensitive' and shown in lighter shades and those from non-sensitive sessions labeled ``Other'' and shown in darker shade.   Separate subplots are shown for each of the first 6 probe query response pages.  The points indicated by the arrows in each subplot show the average of the measured values for the sensitive and non-sensitive user sessions.     

Proposition \ref{prop:1} tell us that that the average of the elements of $\VEC{P}$ is $1$ when topics are a priori equiprobable, i.e. $\VEC{P}$ lies on the line $x+y=2$.    For reference, the line $x+y=2$ is marked on Figure~\ref{fig:base:random:pri}, as well as the mid-point value $(1.0,1.0)$.  It can be seen that the measured data is clustered around this line, as expected due to our choice of regularisation parameter $\lambda$.   As the probe step increases, the $\VEC{P}$ values are increasingly concentrated to the right of $(1.0,1.0)$ for the non-sensitive user sessions and to the left of $(1.0,1.0)$ for the sensitive user sessions.  By the sixth probe the values for the sensitive and non-sensitive sessions have separated into distinct clusters and, as discussed in more detail shortly, we can expect that classification based on the $\VEC{P}$ values will be highly accurate.


\subsection{Detection with \PRIEPS}
\label{sec:classify:prieps}

For detection we normalise the \PRIEPS score $\VEC{P}$ so that each component has mean zero and unit variance.   Namely, for each topic $c_i$, $i=0,1,\dots,N$ we evaluate vector $\VEC{P}$ for each page labelled with $c_i$ in the training set and calculate the empirical mean $\mu_{ij}$ and variance $\sigma_{ij}$ of component $j$ of $\VEC{P}$, $j=1,2,\cdots,N$.   Given a measured score $\VEC{P}$ we calculate the components of the corresponding normalised score $\VEC{Z}$ as
\begin{align}
	Z_i &=  \sum_{j=0}^{N} \left( \frac{P_{j} - \mu_{ij}}{\sigma_{ij} } \right)^{2}, \, i=0, \ldots N
	\label{eqn:z:score:1}
\end{align}
That is, $Z_i$ is the square of the distance of $\VEC{P}$ from the mean vector $\VEC{\mu}_i$ for topic $c_i$ scaled by the components of the corresponding variance vector $\VEC{\sigma}_i$.   

From now on we use this normalised \PRIEPS score, unless otherwise noted, and simply refer to it as the \PRIEPS score.  In our experiments, we declare evidence for profiling of topic $c_I$ has been detected when
\begin{align}
\label{eqn:min:pri:1}
	I \in \underset{0 \leq i \leq N}{ \arg \min} \{Z_0, Z_1, \ldots Z_N \}
\end{align}
That is, we insert probe queries into a user session and use \eqref{eqn:min:pri:1} to decide which topic is detected for each individual probe query in that session.  When deciding which topic is detected, the default tie-breaking action is to default to the non-sensitive ``Other'' category $c_0$.  This happens in practice when, for example, no adverts appear on the result page to be tested. Consequently there is a slight tendency towards classification as ``Other'' so that True Positive ``Sensitive'' topic detections during experiments should be interpreted as a \emph{minimum detection rate}.
 
%

\section{Experimental Results}
\label{sec:experiment:results}

\subsection{Preliminaries}
\label{sec:topical:queries}
We summarise the elements of the data collection setup from \cite{mac2015don} which we use for this present paper to facilitate comparison of results. In all, eleven sensitive user interest topic categories are selected for the purposes of testing. We take the same sensitive topic categories as \cite{mac2015don} to allow comparison of results with previous findings.  Topic categories are described in Table~\ref{tbl:scripts}.  

Of the eleven ``Sensitive' topics, (i) ten are sensitive categories associated with subjects generally identified as causes of discrimination (medical condition, sexual orientation \emph{etc}) or sensitive personal conditions (gambling addiction, financial problems \emph{etc}), (ii) a further sensitive  topic is  related to London as a specific destination location, providing an obviously interesting yet potentially sensitive topic that a recommender system might track, (iii) the last topic is a catch-all category labeled ``Other''.

\begin{table*}[h]
\rowcolors{2}{gray!25}{white}
\setlength{\tabcolsep}{0.15cm} 
\def\arraystretch{1.6}  
\centering
\caption{Categories and associated keyword terms.\label{tbl:scripts}}
\begin{scriptsize} 
\begin{tabular}{@{}m{0.25\columnwidth}m{0.75\columnwidth}@{}}
\cmidrule[1pt]{1-2}

    \textbf{Category}        & \textbf{Keywords }        \\ 
\cmidrule[1pt]{1-2}       
    \textbf{anorexia}          &  nerves eating disorder body image binge diet weight lose fat  \\ 

	\textbf{bankrupt}      &  bankrupt insolvent bad credit poor credit clear your debts insolvency payday insolvent any purpose quick cash benefits low income   \\             

	\textbf{diabetes} &  diabetes mellitus hyperglycaemia blood sugar insulin resistance  \\                                                                      

	\textbf{disabled}         &  disabled special needs accessibility wheelchair  \\                                                                                    

	\textbf{divorce}           &  divorce separation family law   \\                                                                                                        
	
	\textbf{gambling addiction}      &  uncontrollable  addiction compulsive dependency problem support counselling advice therapist therapy help treatment therapeutic recovery anonymous  \\

	\textbf{gay -- homosexuality}   &  gay queer lesbian homosexual bisexual transgender LGBT dyke queen homo    \\                                                                      

	\textbf{location -- London} &  london england uk     \\                                                                                                                  

	\textbf{payday loan}           &  default unsecured debt consolidate advice payday cheap   \\                                                                        

	\textbf{prostate cancer}       &  prostate cancer PSA male urethra urination   \\                                                                                     

	\textbf{unemployed}       &  job seeker recruit search position cv work employment    \\                                                                                                                                                                         

	\textbf{``Other''}        &  Select the top-$50$ queries on Soovle.com as examples of uninformative queries, excluding terms appearing in ``Sensitive' topics.  \\

\end{tabular}                
\end{scriptsize}
\end{table*}

Queries were generated for each of the topics in Table~\ref{tbl:scripts} using the same method as \cite{mac2015don}. For each category apart from the ``Other'' category, a keyword list is created by extracting associated terms from curated sources including Wikipedia (common terms co-occurring on the category page) and Open Directory Project (pages and sub-topics associated with a category).  The extracted keywords are detailed in Table~\ref{tbl:scripts}.   Candidate search queries are then generated for each category by drawing groups of one or more keywords uniformly at random with replacement from the keyword lists.  These candidate queries are manually augmented with common words (and, of \emph{etc}) to yield queries resembling the English language.  In this way a keyword such as fat, for example, might be transformed into a query ``why am i so fat''. Non-sensical or overly robotic queries are removed by manual inspection.

Queries for the ``Other'' topic, which captures topics that are not considered sensitive by the user, are generated from the top query list on the Soovle website \cite{soovle}. Soovle is a meta-search engine that aggregates daily popular query lists from Amazon, Ask, Google, Yahoo and YouTube. It was chosen as a convenient third party source of generic queries. Queries containing terms appearing in the sensitive topic queries are removed from this list and the queries selected for each user session are then drawn uniformly at random.

To construct sequences of queries for use in test sessions, we select a \emph{probe query}, providing a predefined sampling point for data collection. Probe query selection was discussed in detail in \cite{mac2015don}. The probe query ``help and advice'' is used in our experiments for non-medical topics and the ``symptoms and causes'' probe query for medical topics.  The predefined probe query is repeated multiple times with $3-5$ topic specific queries between each instance of the probe query.  In this way ``scripts'' are obtained for each of the $12$ topics in Table~\ref{tbl:scripts}. A query script typically consists of $15-25$ queries including the predefined  probe queries. A test session consists of a single iteration of a single script run from beginning to end.

In our experiments, when implementing the ``Proxy Topic'' defence model, we choose three uninteresting, proxy topics likely to attract adverts, namely tickets for music concerts, searching for bargain vacations and buying a new car. 

\begin{table}[bh]
\centering
\caption{Example query script. The command \texttt{!wait n} instructs the Python script to wait $N$ seconds}
\label{tbl:verbatim:location}
\begin{scriptsize}
\begin{multicols}{3}
\begin{verbatim}
! keywords: addiction dependency 
! probe: help and advice
help and advice
! wait 7
bookies near me
! wait 19
paddy power near me 
! wait 13
bookmaker near me
! wait 4
help and advice
! wait 8
gambling problem
! wait 20
am i a gambling addict 
! wait 15
gambling addiction
! wait 17
help and advice
! wait 2
gambling addiction stories
! wait 11
talk to about gambling 
! wait 13
help and advice
! wait 1
my family does not know i gamble
! wait 8
how to tell my family i gamble
! wait 9
help to stop gambling
! wait 5
stop gambling now
! wait 20
help and advice\end{verbatim}  
\end{multicols}
\end{scriptsize}
\end{table}%

Data was collected using two Linux virtual machines located in a University domain supporting approximately 9,000 users. 
Custom scripts were written to automate query execution and response collection.  These scripts were automated using Python $v3.5.1$, BeautifulSoup4 $v4.4.1$  for HTML processing and phantomjs $v2.1.1$ for browser automation.  
Analysis of results was performed on a $2.3$ GHz Intel Core i7 MacBook Pro using Python $v3.5.1$ to automate data collection and processing. The Python Scipy toolkit $v0.17$ \cite{scikit-learn} was used for text preprocessing.  Numeric processing was performed using the NumPy $v1.10.4$ numerical processing tookit \cite{Idris:2012:NC:2464698}. Plots were created using Matplotlib $v1.5.1$.
 
Measurements of \PRIEPS scores are made for the first $5$ probe queries in a test session during experiments to provide a consistent sample for analysis. A query script selected uniformly at random from one of the topics in Table~\ref{tbl:scripts} was executed each hour over a period of six weeks both as an anonymous user and by logging on as registered Google users. To reduce the appearance of robotic interaction, the script automation program inserts a random pause of 1 to 10 seconds between queries, see Table~\ref{tbl:verbatim:location} for an example. After remaining $5-10$ seconds on a clicked link page, the browser ``back'' button is invoked  to navigate back to the search result page.

As in \cite{mac2015don}, a number of precautions were taken to minimise interactions between runs of each script -- cleaning cookies, history and cache before and after scripts, terminating the session and logging the user out, and waiting for a minimum of twenty minutes between runs to ensure connections  are reset or timed out.  All scripts were run for $3$ registered users and $1$ anonymous user on the Google recommender system, yielding a data set consisting of $21,861$ probe queries in total representing all of the test topics described in Table~\ref{tbl:scripts}. 

Test data was gathered into test data sets based on sensitive topics and types of noise injected and with each test data set consisting of approximately $1,000$ probe queries. A separate hold-back was created for a common training data set of approximately $1,000$ queries. All queries in a test session were automatically labelled with the intended topic of the test session as given by the query script used. For example, all queries from a session about ``prostate`` are labeled as ``prostate'' including probe queries.  In this respect the labels capture intended behaviour of queries, rather than attempting an individual interpretation of specific query keywords during a user session. Test data is automatically divided into $7$ folds for processing so that, when presented, reported averages and standard errors are taken over $7$ distinct, randomised sub-samples of test data unless otherwise noted.

 The \PRIEPS estimator in \eqref{eqn:mk001} uses the training dataset as the source of prior knowledge.  During testing we will use the same training data set in all tests for consistent comparison. We do not re-train \PRIEPS during testing as new adverts are encountered. Experimental measurements of \PRIEPS are with respect to this common training set. The experimental value of \PRIEPS for the $k^{th}$ query step also measures the cumulative effects of events for the previous $k-1$ steps.
  
   \begin{table}[h] 
\setlength{\tabcolsep}{0.1cm}
\centering
\caption{Estimated a priori probability distribution for test topics.}
\label{tbl:Test:apriori}
\begin{tabular}{L{3cm}L{2cm}L{3cm}L{2cm}}
\cmidrule[1pt]{1-4}
\textbf{Topic ($c_i$)} & \textbf{$\hat{p}_i$} & \textbf{Topic ($c_i$)} & \textbf{$\hat{p}_i$}    \\ \cmidrule[1pt]{1-4}
anorexia   & 0.0829 & gay        & 0.0835 \\
bankrupt   & 0.0833 & location   & 0.0836 \\
diabetes   & 0.0825 & payday     & 0.0836 \\
disabled   & 0.0836 & prostate   & 0.0826 \\
divorce    & 0.0836 & unemployed & 0.0836  \\
gambling   & 0.0836 & other      & 0.0834 \\\cmidrule[1pt]{1-4}

\multicolumn{2}{r}{\textbf{RMSE}  (with respect to equally likely prior distribution):}  & \multicolumn{2}{l}{0.0004}\\ 
\cmidrule[1pt]{1-4}
\end{tabular}
\end{table}%

To estimate the prior probability distribution of topics during experiments, we solve \eqref{eqn:prior:estimator} for $\lambda$.  Using the training set data, a value of $\lambda \leq 0.01$ solved \eqref{eqn:prior:estimator} with an RMSE of $0.04\%$. The resulting estimated a priori distribution for all test topics is shown in Table~\ref{tbl:Test:apriori}.

\subsection{Comparison with Previous Results}
\label{sec:compare:previous}
We compare the detection capability of the decision function in \eqref{eqn:min:pri:1} with previous results obtained in \cite{mac2015don}. Results in Table~\ref{tbl:old:new}(a) were produced by processing data from \cite{mac2015don} with the \PRIEPS estimator and then using \eqref{eqn:min:pri:1} to decide which topic is detected. We declare a topic $c_i$ has been detected during a query session, consisting of $5$ probe queries, if \emph{at least one} of the $5$ probe queries is detected as topic $c_i$ to allow direct comparison with previous results. For comparison, detection results for the \PRI estimator from  \cite{mac2015don}\footnote{Taken from Table XIV(b) in \cite{mac2015don}}, are reproduced here as Table~\ref{tbl:old:new}(b). 

\begin{table}[h]
\captionsetup{position=bottom}
\caption{Comparison of measured detection rate of \emph{at least one} individual ``Sensitive' topic in a session of $5$ probes.\label{tbl:old:new}}
\setlength{\tabcolsep}{0.18cm}
\begin{scriptsize}

\subfloat[\PRIEPS]{

        \begin{tabular}{@{}lccccccccccc@{}}
        \cmidrule[1pt]{2-12}
\textbf{\shortstack[l]{Reference\\Topic\\{}}}   & \rot{\textbf{anorexia}}  & \rot{\textbf{bankrupt}}   & \rot{\textbf{diabetes}}  & \rot{\textbf{disabled}} & \rot{\textbf{divorce}} &  \rot{\textbf{gambling}} & \rot{\textbf{gay}}   & \rot{\textbf{location}} & \rot{\textbf{payday}} & \rot{\textbf{prostate}} & \rot{\textbf{unemployed}~} \\ \cmidrule[1pt]{1-12}
      \multicolumn{1}{l}{\textbf{True Detect}}      & 100.0\% & 100.0\% & 100.0\% & 100.0\% & 100.0\% & 100.0\% & 100.0\%  & 100.0\% & 100.0\% & 100.0\%  & 100.0\% \\
        \multicolumn{1}{l}{\textbf{False Detect}}     & 0.0\%   &0.0\% & 0.0\%  & 0.0\% & 0.0\%   & 0.0\%   & 0.0\%   & 0.0\%  & 0.0\%   & 0.0\%   & 0.0\%   \\ \cmidrule[1pt]{1-12}                \end{tabular}        
}\\
\subfloat[\PRI]{

        \begin{tabular}{@{}lccccccccccc@{}}
        \cmidrule[1pt]{2-12}
\textbf{\shortstack[l]{Reference\\Topic\\{}}}   & \rot{\textbf{anorexia}}  & \rot{\textbf{bankrupt}}   & \rot{\textbf{diabetes}}  & \rot{\textbf{disabled}} & \rot{\textbf{divorce}} &  \rot{\textbf{gambling}} & \rot{\textbf{gay}}   & \rot{\textbf{location}} & \rot{\textbf{payday}} & \rot{\textbf{prostate}} & \rot{\textbf{unemployed}~} \\ \cmidrule[1pt]{1-12}
\multicolumn{1}{l}{\textbf{True Detect}}  & 97.0\%  & 100.0\% & 100.0\% & 100.0\% & 100.0\% & 100.0\% & 100.0\% & 100.0\% & 100.0\% & 100.0\% & 100.0\% \\
\multicolumn{1}{l}{\textbf{False Detect}} & 4.0\%   & 0.0\%   & 8.0\%   & 0.0\%   & 0.0\%   & 0.0\%   & 0.0\%   & 0.0\%   & 0.0\%   & 0.0\%   & 0.0\%   \\ \cmidrule[1pt]{1-12}
        \end{tabular}
}               
\end{scriptsize}
\end{table}%

In Table~~\ref{tbl:old:new}(b) ``Sensitive' topic correct detection rates of $97 - 100\%$ per query session are reported using the techniques in \cite{mac2015don}. Noise from incorrect detections, where test scripts labelled as ``Other'' are incorrectly detected as ``Sensitive' lie between $0 - 8\%$, depending on topic. 

Applying the same detection methodology and data as \cite{mac2015don}, Table~~\ref{tbl:old:new}(a) shows that True Detection rates using \PRIEPS estimator and \eqref{eqn:min:pri:1} are better or equal for each topic than the rates reported in \cite{mac2015don} in Table~\ref{tbl:old:new}(b). The False Detection rates are also better or equal in the case of all topics tested. Overall the results obtained using  \PRIEPS estimator and \eqref{eqn:min:pri:1} compare favourably with the results obtained in \cite{mac2015don}.

For the remainder of this paper, we use \eqref{eqn:min:pri:1} as a decision function for detection unless otherwise noted.

\subsection{Establishing a Test Baseline}
\label{sec:test:setup}

We begin by considering a setup where we submit sequences of queries, interleaved with probe queries, in what we term a ``no click, no noise'' model where there is no injected noise and where no items are clicked on any of the search results pages. This model provides a baseline where the queries alone are available to the recommender from which to learn about a user session as it progresses. We will compare further results against this baseline. 

Table~\ref{tbl:distinguish:baseline} shows the True Detection and False Detection rates for each sensitive topic considered. Results in Table~\ref{tbl:distinguish:baseline} are calculated by taking seven distinct, randomised test folds from the relevant test data set. Note that the detection rates reported in this section, and from now on, are for the \PRIEPS scores for individual probe queries.   Also, we note that we use the shorthand ``$95\%$ True Detection'' for $95\%$ of \emph{individual probe queries} were correctly detected as ``Sensitive' given the reference test topic was ``Sensitive'.   This is in contrast to the previous section where the detection rates are calculated on the basis of five probe queries, used to allow comparison with \cite{mac2015don}.  Given the detection rate $p$ for an individual probe we can estimate the rate based on five probes as $1-(1-p)^5$, e.g. when $p=85\%$ then $1-(1-p)^5=99.99\%$.

\begin{table*}[h]
\centering{
\setlength{\tabcolsep}{0.02cm}
\begin{scriptsize}
\caption{Baseline measured detection rate of individual ``Sensitive' topic probe queries (Averaged over $7$ test data folds).\label{tbl:distinguish:baseline}}
        \begin{tabular}{@{}L{2cm}C{1cm}C{1cm}C{1cm}C{1cm}C{1cm}C{1cm}C{1cm}C{1cm}C{1cm}C{1cm}C{1cm}|C{1cm}@{}}
        \cmidrule[1pt]{1-13}

\textbf{Reference Topic} & \rot{\textbf{anorexia}} & \rot{\textbf{bankrupt}} & \rot{\textbf{diabetes}} & \rot{\textbf{disabled}} & \rot{\textbf{divorce}} & \rot{\textbf{gambling}} & \rot{\textbf{gay}} & \rot{\textbf{location}} & \rot{\textbf{payday}} & \rot{\textbf{prostate}} & \rot{\textbf{unemployed}} & \rot{\textbf{Average}}  \\ \cmidrule[1pt]{1-13}

\textbf{True Detect}    & 100.0\%  & 100.0\%  & 100.0\%  & { }82.9\%  & 100.0\%  & { }94.8\%  & 100.0\%  & { }93.1\%  & { }93.8\%  & 100.0\%  & { }92.7\%  & 96.1\%  \\

\textbf{False Detect}   & { }0.0\%  & { }0.0\%  & { }0.0\%  & { }0.0\%  & { }0.0\%  & { }0.0\%  & { }0.0\%  & { }0.0\%  & { }0.0\%  & { }0.0\%  & { }0.0\%  & 0.0\%  \\ \cmidrule[0.5pt]{1-13}

\textbf{SEM Range}  & & \multicolumn{5}{l}{ \textbf{False Detect:}  0.0\% - 0.0\%}  & & \multicolumn{5}{l}{\textbf{True Detect:}  0.0\% - 0.5\%, } \\ \cmidrule[0.5pt]{1-13}
        \end{tabular}        
\end{scriptsize}
}
\end{table*}%

The baseline results in Table~\ref{tbl:distinguish:baseline} indicate that testing for \emph{any} ``Sensitive' topic (True Detection Rate $= 96\%$, in the final column) will result in True Detection at least 4 out of 5 times in a standard query session with probability $>95\%$.  The same holds at individual ``Sensitive' topic level for all topics apart from ``Disabled'' where True Detection occurs at least 2 out of 5 times with probability $95\%$.

\begin{table*}[h]
\centering{
\setlength{\tabcolsep}{0.02cm}
\begin{scriptsize} 
\caption{Baseline lower bounds on mutual plausible deniability ($\delta^{*}$) estimated from maximum and minimum measured values of \PRIEPS\label{tbl:plausible:detect:baseline}}
\begin{tabular}{@{}L{2cm}C{1cm}C{1cm}C{1cm}C{1cm}C{1cm}C{1cm}C{1cm}C{1cm}C{1cm}C{1cm}C{1cm}|C{1cm}@{}}
\cmidrule[1pt]{1-13}

\textbf{Reference Topic(c)} & \rot{\textbf{anorexia}} & \rot{\textbf{bankrupt}} & \rot{\textbf{diabetes}} & \rot{\textbf{disabled}} & \rot{\textbf{divorce}} & \rot{\textbf{gambling}} & \rot{\textbf{gay}} & \rot{\textbf{location}} & \rot{\textbf{payday}} & \rot{\textbf{prostate}} & \rot{\textbf{unemployed}} & \rot{\textbf{Average}}  \\ \cmidrule[1pt]{1-13} \\

\cmidrule[1pt]{1-13} 

\textbf{$\delta^{*}$}   & 0.10 & 0.18  & 0.10  & 0.18  & 0.18  & 0.18  & 0.12  & 0.18  & 0.18 & 0.12  & 0.18  & 0.16  \\ \cmidrule[0.5pt]{1-13}

\textbf{SEM Range}  & &  & & \multicolumn{5}{l}{\textbf{$\delta^{*}$: }  0.0 - 0.0} \\ \cmidrule[0.5pt]{1-13}

\end{tabular}    

\end{scriptsize}
}
\end{table*}%

Turning to Mutual Plausible Deniability, we apply expression~\eqref{eqn:mpd:3} to estimate the \emph{best possible} mutual plausible deniability $\delta^*$ available to the user for each topic.   To do this we estimate $\delta^*$ as the maximum of $\frac{\FPRIHAT^*-\FPRIHATMIN}{N+1}$ over the response to all probe queries in a session, where $\FPRIHAT^*$ is the maximum estimated \PRIEPS score on page $k$ and $\FPRIHATMIN$ the minimum.  In Table~\ref{tbl:plausible:detect:baseline} estimates of $\delta^*$ are shown based on the maximum measured values of \PRIEPS during testing for the ``no click, no noise'' model. 
Results are reported for the plausible deniability of interest in a topic $c$ \emph{given all other} topics -- ``Sensitive'' topics plus the ``Other'' topic. In this case, $\delta^{*}$ indicates how well a user may plausibly deny their interest in topic $c$ versus interest in \emph{all other} topics being tested.
The average value of $\delta^{*}$ in Table~\ref{tbl:plausible:detect:baseline} means that on average, there is a difference of $16\%$ between the posterior probabilities that a user is interested in topic $c$ versus having an interest in any other topic. We will use this value to compare the effectiveness of defences of plausible deniability in reducing the values of $\delta^{*}$.   


\subsection{The Effect of Random Noise Injection}
\label{sec:testing:noise}

Following from Section~\ref{sec:by:design}, we now consider the impact of injecting non-informative queries chosen at random from our popular query list into a user session. We simply refer to these as ``random noise'' queries. We consider three levels of random noise queries for testing purposes:

\begin{table*}[h]
\centering{
\captionsetup{position=bottom}
\caption{Comparison of measured detection rate for various noise models (all with ``No Clicks'', averages over $7$ test folds).\label{tbl:distinguish:noise}}
\setlength{\tabcolsep}{0.02cm}
\begin{scriptsize} 
    \subfloat[Low Noise] {
        \centering
\begin{tabular}{@{}L{2cm}C{1cm}C{1cm}C{1cm}C{1cm}C{1cm}C{1cm}C{1cm}C{1cm}C{1cm}C{1cm}C{1cm}|C{1cm}@{}}
\cmidrule[1pt]{1-13}

\textbf{Reference Topic} & \rot{\textbf{anorexia}} & \rot{\textbf{bankrupt}} & \rot{\textbf{diabetes}} & \rot{\textbf{disabled}} & \rot{\textbf{divorce}} & \rot{\textbf{gambling}} & \rot{\textbf{gay}} & \rot{\textbf{location}} & \rot{\textbf{payday}} & \rot{\textbf{prostate}} & \rot{\textbf{unemployed}} & \rot{\textbf{Average}}  \\ \cmidrule[1pt]{1-13} \\

\cmidrule[1pt]{1-13} 

\textbf{True Detect}    & 100.0\%  & 100.0\%  & 100.0\%  & { }94.3\%  & { }94.0\%  & 100.0\%  & 100.0\%  & 100.0\%  & 100.0\%  & 100.0\%  & 100.0\%  & 98.9\%  \\

\textbf{False Detect}   & { }0.0\%  & { }0.0\%  & { }0.0\%  & { }0.0\%  & { }0.0\%  & { }0.0\%  & { }0.0\%  & { }0.0\%  & { }0.0\%  & { }0.0\%  & { }0.0\%  & 0.0\%  \\ \cmidrule[0.5pt]{1-13}

\textbf{SEM Range}  & & \multicolumn{5}{l}{ \textbf{False Detect:}  0.0\% - 0.0\%}  & & \multicolumn{5}{l}{\textbf{True Detect:}  0.0\% - 0.2\%, } \\ \cmidrule[0.5pt]{1-13}        \end{tabular}        
        \label{tbl:low:noise} 
        }\\ 
               
     \subfloat[Medium Noise] {
       \centering 
\begin{tabular}{@{}L{2cm}C{1cm}C{1cm}C{1cm}C{1cm}C{1cm}C{1cm}C{1cm}C{1cm}C{1cm}C{1cm}C{1cm}|C{1cm}@{}}
\cmidrule[1pt]{1-13}

\textbf{True Detect}    & 100.0\%  & 100.0\%  & 100.0\%  & { }95.2\%  & { }92.4\%  & { }94.1\%  & 100.0\%  & 100.0\%  & { }96.2\%  & 100.0\%  & 100.0\%  & 98.0\%  \\

\textbf{False Detect}   & { }0.0\%  & { }0.0\%  & { }0.0\%  & { }0.0\%  & { }0.0\%  & { }0.0\%  & { }0.0\%  & { }0.0\%  & { }0.0\%  & { }0.0\%  & { }0.0\%  & 0.0\%  \\ \cmidrule[0.5pt]{1-13}

\textbf{SEM Range}  & & \multicolumn{5}{l}{ \textbf{False Detect:}  0.0\% - 0.0\%}  & & \multicolumn{5}{l}{\textbf{True Detect:}  0.0\% - 0.1\%, } \\ \cmidrule[0.5pt]{1-13}
        \end{tabular}
        \label{tbl:med:noise}
        } \\
      \subfloat[High Noise] {
       \centering 
\begin{tabular}{@{}L{2cm}C{1cm}C{1cm}C{1cm}C{1cm}C{1cm}C{1cm}C{1cm}C{1cm}C{1cm}C{1cm}C{1cm}|C{1cm}@{}}
\cmidrule[1pt]{1-13}
        

\textbf{True Detect}    & 100.0\%  & 100.0\%  & 100.0\%  & 100.0\%  & 100.0\%  & { }95.9\%  & { }98.9\%  & 100.0\%  & 100.0\%  & 100.0\%  & 100.0\%  & 99.5\%  \\

\textbf{False Detect}   & { }0.0\%  & { }0.0\%  & { }0.0\%  & { }0.0\%  & { }0.0\%  & { }0.0\%  & { }0.0\%  & { }0.0\%  & { }0.0\%  & { }0.0\%  & { }0.0\%  & 0.0\%  \\ \cmidrule[0.5pt]{1-13}

\textbf{SEM Range}  & & \multicolumn{5}{l}{ \textbf{False Detect:}  0.0\% - 0.0\%}  & & \multicolumn{5}{l}{\textbf{True Detect:}  0.0\% - 0.1\%, } \\ \cmidrule[0.5pt]{1-13}
        \end{tabular}
        \label{tbl:high:noise}
        }                              
\end{scriptsize}
}
\end{table*}%
\begin{description}[labelsep=1em, itemsep=2pt,parsep=2pt]
	\item [``Low Noise''] The automation scripts select uninteresting queries uniformly at random from the  top-query list and inject a single random noise query after every topic-specific query so that the ``signal-to-noise ratio'' of sensitive to noise queries in this case is $1:1$.
	\item [``Medium Noise''] Here the automation scripts inject two randomly selected  queries after each topic-specific query for a signal to noise ration of $1:2$.
	\item [``High Noise''] In this noise-model with the highest noise setting, three random noise queries are injected, resulting in a signal-to-noise ratio of $1:3$.
\end{description} 

Note also that the automation scripts were configured to ensure the relevant number of noise queries was always injected \emph{immediately before} each probe query.   Our intention was to construct a ``worst case'' for detection of learning, where probe queries are always separated from sensitive user queries by the specified number of noise queries.

\begin{table*}[h]
\centering{
\captionsetup{position=bottom}
\caption{Lower bounds on plausible deniability ($\delta(c)$) estimated from maximum and minimum measured values of \PRIEPS for various random query noise models\label{tbl:plausible:detect:noise}}
\setlength{\tabcolsep}{0.02cm}
\begin{scriptsize} 
      \subfloat[Low Noise] {
        \centering
\begin{tabular}{@{}L{2cm}C{1cm}C{1cm}C{1cm}C{1cm}C{1cm}C{1cm}C{1cm}C{1cm}C{1cm}C{1cm}C{1cm}|C{1cm}@{}}
\cmidrule[1pt]{1-13}

\textbf{Reference Topic(c)} & \rot{\textbf{anorexia}} & \rot{\textbf{bankrupt}} & \rot{\textbf{diabetes}} & \rot{\textbf{disabled}} & \rot{\textbf{divorce}} & \rot{\textbf{gambling}} & \rot{\textbf{gay}} & \rot{\textbf{location}} & \rot{\textbf{payday}} & \rot{\textbf{prostate}} & \rot{\textbf{unemployed}} & \rot{\textbf{Average}}  \\ \cmidrule[1pt]{1-13} \\

\cmidrule[1pt]{1-13}

\textbf{$\delta^{*}$}   & 0.08  & 0.18  & 0.12  & 0.18  & 0.18  & 0.18  & 0.12  & 0.18  & 0.18 & 0.10  & 0.18  & 0.16  \\ \cmidrule[0.5pt]{1-13}

\textbf{SEM Range}  & &   & & \multicolumn{5}{l}{\textbf{$\delta^{*}$: }  0.0 - 0.01 } \\ \cmidrule[0.5pt]{1-13}

        \end{tabular}
        \label{tbl:plausible:detect:low:noise}
        } \\

     \subfloat[Medium Noise] {
       \centering 
\begin{tabular}{@{}L{2cm}C{1cm}C{1cm}C{1cm}C{1cm}C{1cm}C{1cm}C{1cm}C{1cm}C{1cm}C{1cm}C{1cm}|C{1cm}@{}}
\cmidrule[1pt]{1-13}

\textbf{$\delta^{*}$}   & 0.14  & 0.18  & 0.12  & 0.18  & 0.18  & 0.18  & 0.18  & 0.18  & 0.18 & 0.10  & 0.18  & 0.16  \\ \cmidrule[0.5pt]{1-13}

\textbf{SEM Range}  & & & & \multicolumn{5}{l}{\textbf{$\delta^{*}$: }  0.0 - 0.01 } \\ \cmidrule[0.5pt]{1-13}

        \end{tabular}
        \label{tbl:plausible:detect:med:noise}
        } \\
      \subfloat[High Noise] {
       \centering 
\begin{tabular}{@{}L{2cm}C{1cm}C{1cm}C{1cm}C{1cm}C{1cm}C{1cm}C{1cm}C{1cm}C{1cm}C{1cm}C{1cm}|C{1cm}@{}}
\cmidrule[1pt]{1-13}
        

\textbf{$\delta^{*}$}   & 0.06  & 0.18  & 0.10  & 0.18  & 0.18  & 0.18  & 0.18  & 0.18  & 0.18 & 0.08  & 0.18  & 0.18  \\ \cmidrule[0.5pt]{1-13}

\textbf{SEM Range}  & & & & \multicolumn{5}{l}{\textbf{$\delta^{*}$: }  0.0 - 0.01 } \\ \cmidrule[0.5pt]{1-13}
        \end{tabular}
        \label{tbl:plausible:detect:high:noise}
        }                              
\end{scriptsize}
}
\end{table*}%

Table~\ref{tbl:distinguish:noise} shows measured True Detection and False Detection rates as the degree of noise injected is varied.   As before, links in the response pages are not clicked on and this data can be directly compared against the baseline in Table~\ref{tbl:distinguish:baseline}. 
It can be seen that none of the tested random noise models provides consistent reduction in the True Detection Rate for all topics, with this rate remaining high for the majority of topics across all noise-models.  The results in Table~\ref{tbl:distinguish:noise} indicate that even the ``High Random Noise'' model fails to reduce the True Detection Rate of \PRIEPS. False detections are also low for all noise-models suggesting that accuracy is also high.

In Table~\ref{tbl:plausible:detect:noise} estimates of the lower bound on mutual plausible deniability are reported based on the minimum and maximum measured values of \PRIEPS during testing. Comparing Table~\ref{tbl:plausible:detect:noise} with the baseline results in Table~\ref{tbl:plausible:detect:baseline}, injecting random noise has not improved user mutual plausible deniability significantly on average. We conclude that injection of random noise, even at substantial levels, does not provide a useful defence for mutual plausible deniability in our experiments. 

\subsection{The Effect of Click Strategies}
\label{sec:testing:clicks}

We now consider whether it is possible to disrupt search engine learning by careful clicking of the links on response pages.  

Intuitively, from the search engine's point of view clicking on links is a form of active feedback by a user and so potentially informative of user interests.  This is especially true when, for example, a user is carrying out exploratory search where their choice of keywords is not yet well-tuned to their topic of interest.  Previous studies have also indicated that there is good reason to believe that user clicks on links are an important input into recommender system learning. In \cite{mac2015don} (Section $6.4$), user clicks emulated using the ``Click Relevant'' click-model were reported to result in increases of $60\%$ -- $450\%$ in the advert content, depending on the ``Sensitive'  topic tested. 

We consider four different click strategies to emulate a range of user click behaviours:

\begin{description}
	\item [``No Click''] No items are clicked on in the response page to a query. This user click-model does not provide additional user preference information to the recommender system due to click behaviour. This click model is used in the baseline measurements presented in Sections~\ref{sec:test:setup}.
	\item [``Click Relevant''] Given the response page to a query, for each search result and advert we calculate the Term-Frequency (TF) of the visible text with respect to the keywords associated with the test session topic of interest, see Table~\ref{tbl:scripts}. When $TF > 0.1$ for an item, the item is clicked, otherwise it is not clicked. This user click-model provides relevant feedback to the recommender system about the information goal of the user.
	\item [``Click Non-relevant''] TF is calculated for each item with respect to the category of interest for the session in question as for the ``Click Relevant'' click-model, \emph{except} that items are clicked when the TF score is below the threshold and so they are deemed non-relevant to the topic, that is when $TF \leq 0.1$. This user click-model attempts to confuse the recommender system by providing feedback that is not relevant to the true topic of interest to the user.
	\item [``Click All''] All items on the response page for a query are clicked. This user click-model gives the recommender system a ``noisy'' click signal, including clicks on items relevant and non-relevant to the user's information goal.
	\item [``Click 2 Random Items''] Two items appearing on the response page for a query are selected uniformly at random with replacement and clicked.  
\end{description} 

In all cases, when uninteresting, ``noise'' queries are included in a query session, the relevant user click-strategy is also applied to the result pages of these queries. In this way we avoid providing an obvious signal to the recommender system to differentiate uninteresting queries from queries related to ``Sensitive' topics.  Items on the result page in response to probe queries are \emph{not} clicked so that the probe query does not provide any additional information to the recommender system.

Table~\ref{tbl:distinguish:clicks} shows measured True Detection and False Detection rates for each strategy.  It can be seen that the True Detection rate remains high for all click strategies tested, while the False Detection rate remains consistently low, indicating that detection remains accurate.  

\begin{table*}[h]
\captionsetup{position=bottom}
\caption{Comparison of measured detection rate for various click models, all with ``no noise'', averages over $7$ test folds.\label{tbl:distinguish:clicks}}
\setlength{\tabcolsep}{0.02cm}
\begin{scriptsize}
     \subfloat[Click Relevant] {
\begin{tabular}{@{}L{2cm}C{1cm}C{1cm}C{1cm}C{1cm}C{1cm}C{1cm}C{1cm}C{1cm}C{1cm}C{1cm}C{1cm}|C{1cm}@{}}
\cmidrule[1pt]{1-13}

\textbf{Reference Topic} & \rot{\textbf{anorexia}} & \rot{\textbf{bankrupt}} & \rot{\textbf{diabetes}} & \rot{\textbf{disabled}} & \rot{\textbf{divorce}} & \rot{\textbf{gambling}} & \rot{\textbf{gay}} & \rot{\textbf{location}} & \rot{\textbf{payday}} & \rot{\textbf{prostate}} & \rot{\textbf{unemployed}} & \rot{\textbf{Average}}  \\ \cmidrule[1pt]{1-13} \\

\cmidrule[1pt]{1-13} 

\textbf{True Detect}    & 100.0\%  & { }98.7\%  & 100.0\%  & { }95.1\%  & 100.0\%  & 100.0\%  & 100.0\%  & 100.0\%  & 100.0\%  & 100.0\%  & 100.0\%  & 99.4\%  \\

\textbf{False Detect}   & { }0.0\%  & { }0.0\%  & { }0.0\%  & { }0.0\%  & { }0.0\%  & { }0.0\%  & { }0.0\%  & { }0.0\%  & { }0.0\%  & { }0.0\%  & { }0.0\%  & 0.0\%  \\ \cmidrule[0.5pt]{1-13}

\textbf{SEM Range}  & & \multicolumn{5}{l}{ \textbf{False Detect:}  0.0\% - 0.0\%}  & & \multicolumn{5}{l}{\textbf{True Detect:}  0.0\% - 0.2\%, } \\ \cmidrule[0.5pt]{1-13}
        
        \end{tabular}        
        \label{tbl:click:relevant} 
        }\\        
     \subfloat[Click Non-relevant] {
       \centering 
        \begin{tabular}{@{}L{2cm}C{1cm}C{1cm}C{1cm}C{1cm}C{1cm}C{1cm}C{1cm}C{1cm}C{1cm}C{1cm}C{1cm}|C{1cm}@{}}
         \cmidrule[1pt]{1-13}


\textbf{True Detect}    & 100.0\%  & { }96.6\%  & 100.0\%  & 100.0\%  & 100.0\%  & { }92.2\%  & 100.0\%  & { }98.1\%  & 100.0\%  & 100.0\%  & 100.0\%  & 98.8\%  \\

\textbf{False Detect}   & { }0.0\%  & { }0.0\%  & { }0.0\%  & { }0.0\%  & { }0.0\%  & { }0.0\%  & { }0.0\%  & { }0.0\%  & { }0.0\%  & { }0.0\%  & { }0.0\%  & 0.0\%  \\ \cmidrule[0.5pt]{1-13}

\textbf{SEM Range}  & & \multicolumn{5}{l}{ \textbf{False Detect:}  0.0\% - 0.0\%}  & & \multicolumn{5}{l}{\textbf{True Detect:}  0.0\% - 0.1\%, } \\ \cmidrule[0.5pt]{1-13}

        \end{tabular}
        \label{tbl:click:non:relevant}
        } \\
      \subfloat[Click All] {
       \centering 
        \begin{tabular}{@{}L{2cm}C{1cm}C{1cm}C{1cm}C{1cm}C{1cm}C{1cm}C{1cm}C{1cm}C{1cm}C{1cm}C{1cm}|C{1cm}@{}}
        \cmidrule[1pt]{1-13} 
        

\textbf{True Detect}    & 100.0\%  & 100.0\%  & 100.0\%  & { }84.7\%  & { }84.4\%  & { }94.7\%  & { }96.8\%  & { }69.8\%  & 100.0\%  & 100.0\%  & { }73.2\%  & 91.2\%  \\

\textbf{False Detect}   & { }0.0\%  & { }0.0\%  & { }0.0\%  & { }0.0\%  & { }0.0\%  & { }0.0\%  & { }0.0\%  & { }0.0\%  & { }0.0\%  & { }0.0\%  & { }0.0\%  & 0.0\%  \\ \cmidrule[0.5pt]{1-13}

\textbf{SEM Range}  & & \multicolumn{5}{l}{ \textbf{False Detect:}  0.0\% - 0.0\%}  & & \multicolumn{5}{l}{\textbf{True Detect:}  0.0\% - 1.2\%, } \\ \cmidrule[0.5pt]{1-13}

        \end{tabular}
        \label{tbl:click:all}
        }  \\
      \subfloat[Click 2 Random Items] {
       \centering 
        \begin{tabular}{@{}L{2cm}C{1cm}C{1cm}C{1cm}C{1cm}C{1cm}C{1cm}C{1cm}C{1cm}C{1cm}C{1cm}C{1cm}|C{1cm}@{}}
        \cmidrule[1pt]{1-13} 
        

\textbf{True Detect}    & { }99.0\%  & 100.0\%  & 100.0\%  & { }90.7\%  & { }93.0\%  & { }83.9\%  & 100.0\%  & 100.0\%  & { }94.7\%  & 100.0\%  & { }93.4\%  & 95.9\%  \\

\textbf{False Detect}   & { }0.0\%  & { }0.0\%  & { }0.0\%  & { }0.0\%  & { }0.0\%  & { }0.0\%  & { }0.0\%  & { }0.0\%  & { }0.0\%  & { }0.0\%  & { }0.0\%  & 0.0\%  \\ \cmidrule[0.5pt]{1-13}

\textbf{SEM Range}  & & \multicolumn{5}{l}{ \textbf{False Detect:}  0.0\% - 0.0\%}  & & \multicolumn{5}{l}{\textbf{True Detect:}  0.0\% - 0.2\%, } \\ \cmidrule[0.5pt]{1-13}
        \end{tabular}
        \label{tbl:click:random:2}
        }                                     
\end{scriptsize}
\end{table*}%
 
Comparing the baseline ``No Click'' average result in the final column of Table~\ref{tbl:distinguish:baseline} with the  same result in each of the subtables in Table~\ref{tbl:distinguish:clicks} indicate that the ``Click Relevant'' and ``Click Non-relevant'' click-models \emph{increase} the True Detect Rate by $2\%$ -- $3\%$ on average over all ``Sensitive' topics.  That is, these click strategies increase evidence for learning by the search engine.    The ``Click All'' model results in the largest average reduction of $7\%$ indicating that clicking on every item may cause some confusion; note that the ``False Detect'' rate remains low so we do not observe false interest in other sensitive topics as a result of the click model.   Of the eleven sensitive topics tested, three topics have significantly reduced True Detect rates under the ``Click All'' model -- ``gambling'', ``location'' and ``unemployed''.  This suggests that ``Click All'' may provide some protection for some topics, however it is difficult to predict which.    However, the protection is rather weak with the True Detect rates remaining high (greater than about 70\% per individual probe query) for all topics.



\begin{table*}[h]
\centering{
\captionsetup{position=bottom}
\caption{Lower bounds on plausible deniability, $\delta^{*}$, estimated from maximum and minimum measured values of \PRIEPS for various user click models\label{tbl:plausible:detect:clicks}}
\setlength{\tabcolsep}{0.02cm}
\begin{scriptsize} 
      \subfloat[Click Relevant] {
        \centering
\begin{tabular}{@{}L{2cm}C{1cm}C{1cm}C{1cm}C{1cm}C{1cm}C{1cm}C{1cm}C{1cm}C{1cm}C{1cm}C{1cm}|C{1cm}@{}}
\cmidrule[1pt]{1-13}

\textbf{Reference Topic(c)} & \rot{\textbf{anorexia}} & \rot{\textbf{bankrupt}} & \rot{\textbf{diabetes}} & \rot{\textbf{disabled}} & \rot{\textbf{divorce}} & \rot{\textbf{gambling}} & \rot{\textbf{gay}} & \rot{\textbf{location}} & \rot{\textbf{payday}} & \rot{\textbf{prostate}} & \rot{\textbf{unemployed}} & \rot{\textbf{Average}}  \\ \cmidrule[1pt]{1-13} \\

\cmidrule[1pt]{1-13}

\textbf{$\delta^{*}$}   & 0.06  & 0.18  & 0.06  & 0.26  & 0.18  & 0.18  & 0.18  & 0.18  & 0.18 & 0.06  & 0.18  & 0.16  \\ \cmidrule[0.5pt]{1-13}

\textbf{SEM Range}  & &   & & \multicolumn{5}{l}{\textbf{$\delta^{*}$: }  0.0 - 0.0 } \\ \cmidrule[0.5pt]{1-13}

        \end{tabular}
        \label{tbl:plausible:detect:click:relevant}
        } \\

     \subfloat[Click Non-Relevant] {
       \centering 
\begin{tabular}{@{}L{2cm}C{1cm}C{1cm}C{1cm}C{1cm}C{1cm}C{1cm}C{1cm}C{1cm}C{1cm}C{1cm}C{1cm}|C{1cm}@{}}
\cmidrule[1pt]{1-13}

\textbf{$\delta^{*}$}   & 0.12  & 0.20 & 0.12  & 0.22  & 0.18  & 0.22  & 0.12  & 0.22  & 0.20  & 0.12  & 0.22  & 0.18  \\ \cmidrule[0.5pt]{1-13}

\textbf{SEM Range}  & & & & \multicolumn{5}{l}{\textbf{$\delta^{*}$: }  0.0 - 0.0 } \\ \cmidrule[0.5pt]{1-13}

        \end{tabular}
        \label{tbl:plausible:detect:click:non:relevant}
        } \\
        
     \subfloat[Click All] {
       \centering 
\begin{tabular}{@{}L{2cm}C{1cm}C{1cm}C{1cm}C{1cm}C{1cm}C{1cm}C{1cm}C{1cm}C{1cm}C{1cm}C{1cm}|C{1cm}@{}}
\cmidrule[1pt]{1-13}

\textbf{$\delta^{*}$}   & 0.10  & 0.18  & 0.08  & 0.18  & 0.18  & 0.18  & 0.12  & 0.18  & 0.18 & 0.06  & 0.18  & 0.16  \\ \cmidrule[0.5pt]{1-13}

\textbf{SEM Range}  & & & & \multicolumn{5}{l}{\textbf{$\delta^{*}$: }  0.0 - 0.0 } \\ \cmidrule[0.5pt]{1-13}

        \end{tabular}
        \label{tbl:plausible:detect:click:all}
        } \\
        
      \subfloat[Click 2 Random Items] {
       \centering 
\begin{tabular}{@{}L{2cm}C{1cm}C{1cm}C{1cm}C{1cm}C{1cm}C{1cm}C{1cm}C{1cm}C{1cm}C{1cm}C{1cm}|C{1cm}@{}}
\cmidrule[1pt]{1-13}
        

\textbf{$\delta^{*}$}   & 0.10  & 0.12  & 0.12  & 0.20  & 0.18  & 0.18  & 0.12  & 0.18  & 0.18 & 0.10  & 0.20  & 0.16  \\ \cmidrule[0.5pt]{1-13}

\textbf{SEM Range}  & & & & \multicolumn{5}{l}{\textbf{$\delta^{*}$: }  0.0 - 0.0 } \\ \cmidrule[0.5pt]{1-13}
        \end{tabular}
        \label{tbl:plausible:detect:click:random:2}
        }                              
\end{scriptsize}
}
\end{table*}%

For mutual plausible deniability, results are reported in Table~\ref{tbl:plausible:detect:clicks}. The average measured values of $\delta^{*}$ are consistent with the baseline reported in Table~\ref{tbl:plausible:detect:baseline}, suggesting that the click strategies tested are not effective on average at improving user ability to plausibly deny their interest in a topic.  The value of $\delta$ is significantly smaller for the topics \emph{diabetes} and \emph{prostate} for all click models. An inspection of result pages for both of these topics indicates that some of the queries for these topics contained multiple search terms and so are possibly confusing for the search engine. For example, in the case of \emph{prostate}, the query ``should i \textbf{exercise} more with \textbf{prostate cancer}''. In such cases the result was that no adverts were returned resulting in lower \PRIEPS scores over the query session. 

It would appear in summary, that clicks transmit information to the search engine, but not as strongly or as consistently as does a revealing query. Consequently none of the user click-models tested appear to offer adequate and sustained protection from ``Sensitive' topic detection. The ability of \PRIEPS to correctly identify ``Other'' or uninteresting queries for each of the click-models tested suggests that it is possible to filter out uninteresting or noise queries with high accuracy. Of the click-models tested, only the ``Click All'' model provides a comparable level of mutual plausible deniability with the baseline ``No Click'' model, so that if the user must click on items then our experiments suggest the least damaging option to protect plausible deniability is to click all items.   

\subsection{The Effect of Proxy Topics}
\label{sec:proxy:topic}

The next privacy protection strategy we consider is the introduction of proxy topics. In this case sequences of queries, each sequence related to a single uninteresting proxy topic topic, are injected into a user session.   The idea here is that each sequence of queries emulates a user session on a proxy topic, and so hopefully misdirects learning by the search engine of user interests that is based on sessions.  In addition, the results in Section \ref{sec:testing:noise} are relevant here since they suggest that isolated queries on a topic tend not to provoke search engine learning and so the opportunity exists to exploit that.  Namely, by interspersing sensitive queries within sessions on uninteresting proxy topics the isolated sensitive queries will hopefully not provoke learning whereas the uninteresting proxy sessions will and in this way we can misdirect learning by the search engine.


In out tests the following proxy topics are used:
\begin{description}
	\item[\textbf{tickets}] Searching for tickets for events in a local Dublin venue called Croke Park
	\item[\textbf{vacation}] Queries related to a vacation in France such as flights and accommodation.
	\item[\textbf{car}] Searches by a user seeking to trade in and change their car.
\end{description}
and related queries are constructed in the same way as for the sensitive topics in Table \ref{tbl:scripts}. The queries associated with each proxy topic are given in Appendix \ref{sec:backm:proxy:topics:queries}.

\begin{table*}[h]
\captionsetup{position=bottom}
\caption{Comparison of measured detection rate for proxy topics and for various click models, all with ``no noise'', averages over $7$ test folds.\label{tbl:plausible:detect:clicks}}
\setlength{\tabcolsep}{0.02cm}
\begin{scriptsize}
     \subfloat[Click Relevant] {
\begin{tabular}{@{}L{2cm}C{1cm}C{1cm}C{1cm}C{1cm}C{1cm}C{1cm}C{1cm}C{1cm}C{1cm}C{1cm}C{1cm}|C{1cm}@{}}
\cmidrule[1pt]{1-13}

\textbf{Reference Topic} & \rot{\textbf{anorexia}} & \rot{\textbf{bankrupt}} & \rot{\textbf{diabetes}} & \rot{\textbf{disabled}} & \rot{\textbf{divorce}} & \rot{\textbf{gambling}} & \rot{\textbf{gay}} & \rot{\textbf{location}} & \rot{\textbf{payday}} & \rot{\textbf{prostate}} & \rot{\textbf{unemployed}} & \rot{\textbf{Average}}  \\ \cmidrule[1pt]{1-13} \\

\cmidrule[1pt]{1-13} 

\textbf{True Detect}   & { }0.0\%  & { }0.0\%  & { }0.0\%  & { }0.0\%  & { }0.0\%  & { }0.0\%  & { }0.0\%  & { }0.0\%  & { }0.0\%  & { }0.0\%  & { }0.0\%  & 0.0\%  \\

\textbf{False Detect}   & { }0.0\%  & { }0.0\%  & { }0.0\%  & { }0.0\%  & { }0.0\%  & { }0.0\%  & { }0.0\%  & { }0.0\%  & { }0.0\%  & { }0.0\%  & { }0.0\%  & 0.0\%  \\ \cmidrule[0.5pt]{1-13}

\textbf{SEM Range}  & & \multicolumn{5}{l}{ \textbf{False Detect:}  0.0\% - 0.0\%}  & & \multicolumn{5}{l}{\textbf{True Detect:}  0.0\% - 0.2\%, } \\ \cmidrule[0.5pt]{1-13}
        
        \end{tabular}        
        \label{tbl:plausible:detect:click:relevant} 
        }\\        
     \subfloat[Click Non-relevant] {
       \centering 
        \begin{tabular}{@{}L{2cm}C{1cm}C{1cm}C{1cm}C{1cm}C{1cm}C{1cm}C{1cm}C{1cm}C{1cm}C{1cm}C{1cm}|C{1cm}@{}}
         \cmidrule[1pt]{1-13}


\textbf{True Detect}   & { }0.0\%  & { }0.0\%  & { }0.0\%  & { }0.0\%  & { }0.0\%  & { }0.0\%  & { }0.0\%  & { }0.0\%  & { }0.0\%  & { }0.0\%  & { }0.0\%  & 0.0\%  \\ 

\textbf{False Detect}   & { }0.0\%  & { }0.0\%  & { }0.0\%  & { }0.0\%  & { }0.0\%  & { }0.0\%  & { }0.0\%  & { }0.0\%  & { }0.0\%  & { }0.0\%  & { }0.0\%  & 0.0\%  \\ \cmidrule[0.5pt]{1-13}

\textbf{SEM Range}  & & \multicolumn{5}{l}{ \textbf{False Detect:}  0.0\% - 0.0\%}  & & \multicolumn{5}{l}{\textbf{True Detect:}  0.0\% - 0.1\%, } \\ \cmidrule[0.5pt]{1-13}

        \end{tabular}
        \label{tbl:plausible:detect:click:non:relevant}
        } \\
      \subfloat[Click All] {
       \centering 
        \begin{tabular}{@{}L{2cm}C{1cm}C{1cm}C{1cm}C{1cm}C{1cm}C{1cm}C{1cm}C{1cm}C{1cm}C{1cm}C{1cm}|C{1cm}@{}}
        \cmidrule[1pt]{1-13} 
        

\textbf{True Detect}   & { }0.0\%  & { }0.0\%  & { }0.0\%  & { }0.0\%  & { }0.0\%  & { }0.0\%  & { }0.0\%  & { }0.0\%  & { }0.0\%  & { }0.0\%  & { }0.0\%  & 0.0\%  \\

\textbf{False Detect}   & { }0.0\%  & { }0.0\%  & { }0.0\%  & { }0.0\%  & { }0.0\%  & { }0.0\%  & { }0.0\%  & { }0.0\%  & { }0.0\%  & { }0.0\%  & { }0.0\%  & 0.0\%  \\ \cmidrule[0.5pt]{1-13}

\textbf{SEM Range}  & & \multicolumn{5}{l}{ \textbf{False Detect:}  0.0\% - 0.0\%}  & & \multicolumn{5}{l}{\textbf{True Detect:}  0.0\% - 1.2\%, } \\ \cmidrule[0.5pt]{1-13}

        \end{tabular}
        \label{tbl:plausible:detect:click:all}
        }  \\
      \subfloat[Click 2 Random Items] {
       \centering 
        \begin{tabular}{@{}L{2cm}C{1cm}C{1cm}C{1cm}C{1cm}C{1cm}C{1cm}C{1cm}C{1cm}C{1cm}C{1cm}C{1cm}|C{1cm}@{}}
        \cmidrule[1pt]{1-13} 
        

\textbf{True Detect}   & { }0.0\%  & { }0.0\%  & { }0.0\%  & { }0.0\%  & { }0.0\%  & { }0.0\%  & { }0.0\%  & { }0.0\%  & { }0.0\%  & { }0.0\%  & { }0.0\%  & 0.0\%  \\

\textbf{False Detect}   & { }0.0\%  & { }0.0\%  & { }0.0\%  & { }0.0\%  & { }0.0\%  & { }0.0\%  & { }0.0\%  & { }0.0\%  & { }0.0\%  & { }0.0\%  & { }0.0\%  & 0.0\%  \\ \cmidrule[0.5pt]{1-13}

\textbf{SEM Range}  & & \multicolumn{5}{l}{ \textbf{False Detect:}  0.0\% - 0.0\%}  & & \multicolumn{5}{l}{\textbf{True Detect:}  0.0\% - 0.2\%, } \\ \cmidrule[0.5pt]{1-13}
        \end{tabular}
        \label{tbl:plausible:detect:click:random:2}
        }                                     
\end{scriptsize}
\end{table*}%

Proxy topic query scripts where constructed by selecting a sensitive topic of interest from Table \ref{tbl:scripts}, and then selecting an uninteresting proxy topic from the list of $3$ proxy topics. Having decided on a sensitive query we wish to issue, we select at least three and no more than four queries related to the proxy topic from a prepared list of proxy topic queries. We next randomly shuffle the order of the selected sensitive and proxy topic queries. In this way there is always a subgroup of at least two proxy topic queries next to each other in each query session. Finally, for testing purposes, we place a probe query before and after each block of 3-4 proxy + 1 sensitive queries to measure changes in \PRIEPS score. We repeat this exercise using the same proxy topic until a typical query session consisting of $5$ probe queries is created. An example of a typical randomly generated sequence consisting of $5$ probes queries is:

\begin{mdframed}[leftmargin=10pt,rightmargin=10pt]
\BLUE{probe}, \GREEN{proxy}, \RED{sensitive}, \GREEN{proxy}, \GREEN{proxy}, \BLUE{probe},  \RED{sensitive}, \GREEN{proxy}, \GREEN{proxy}, \GREEN{proxy}, \BLUE{probe} \GREEN{proxy}, \GREEN{proxy}, \GREEN{proxy}, \RED{sensitive}, \BLUE{probe}, \GREEN{proxy}, \GREEN{proxy}, \RED{sensitive}, \GREEN{proxy}, \BLUE{probe}
\end{mdframed}

Data was collected for $2,300$ such proxy topic sessions. This included each  of the sensitive topics in Table \ref{tbl:scripts} and each of the click models described in previous sections. The same \PRIEPS setup as before was used, including the same training set, to detect evidence of learning by the search engine. 

Measured detection rates are shown in Table \ref{tbl:plausible:detect:clicks}.  Remarkably, the True Positive detection rate for all topics and for all click-models tested is $0\%$.  That is, we find \emph{no evidence of learning} by the search engine in any of the tests.   Since our detection approach is demonstrated to be notably sensitive to search engine learning in earlier sections, we can reasonably infer that this result is not due to a defect in the detection methodology but rather genuinely reflects successful misdirection of the search engine away from sensitive topics.

\begin{table*}[h]
\centering{
\captionsetup{position=bottom}
\caption{Lower bounds on plausible deniability ($\delta(c)$) estimated from maximum and minimum measured values of \PRIEPS for proxy topic and various click models\label{tbl:proxy:clicks}}
\setlength{\tabcolsep}{0.02cm}
\begin{scriptsize} 
      \subfloat[Proxy Topic, No Click] {
        \centering
\begin{tabular}{@{}L{2cm}C{1cm}C{1cm}C{1cm}C{1cm}C{1cm}C{1cm}C{1cm}C{1cm}C{1cm}C{1cm}C{1cm}|C{1cm}@{}}
\cmidrule[1pt]{1-13}

\textbf{Reference Topic(c)} & \rot{\textbf{anorexia}} & \rot{\textbf{bankrupt}} & \rot{\textbf{diabetes}} & \rot{\textbf{disabled}} & \rot{\textbf{divorce}} & \rot{\textbf{gambling}} & \rot{\textbf{gay}} & \rot{\textbf{location}} & \rot{\textbf{payday}} & \rot{\textbf{prostate}} & \rot{\textbf{unemployed}} & \rot{\textbf{Average}}  \\ \cmidrule[1pt]{1-13} \\

\cmidrule[1pt]{1-13}

\textbf{$\delta^{*}$}   & 0.50  & 0.50    & 0.50    & 0.50  & 0.50  & 0.50  & 0.50  & 0.50  & 0.50 & 50\%  & 0.50  & 0.50  \\ \cmidrule[0.5pt]{1-13}

\textbf{SEM Range}  & &   & & \multicolumn{5}{l}{\textbf{$\delta^{*}$: }  0.0 - 0.0 } \\ \cmidrule[0.5pt]{1-13}

        \end{tabular}
        \label{tbl:plausible:detect:proxy:no:click}
        } \\

      \subfloat[Proxy Topic, Click Relevant] {
        \centering
\begin{tabular}{@{}L{2cm}C{1cm}C{1cm}C{1cm}C{1cm}C{1cm}C{1cm}C{1cm}C{1cm}C{1cm}C{1cm}C{1cm}|C{1cm}@{}}
\cmidrule[1pt]{1-13}

\textbf{$\delta^{*}$}   & 0.50  & 0.50    & 0.50    & 0.50  & 0.50  & 0.50  & 0.50  & 0.50  & 0.50 & 50\%  & 0.50  & 0.50  \\ \cmidrule[0.5pt]{1-13}

\textbf{SEM Range}  & &   & & \multicolumn{5}{l}{\textbf{$\delta^{*}$: }  0.0 - 0.0 } \\ \cmidrule[0.5pt]{1-13}

        \end{tabular}
        \label{tbl:plausible:detect:proxy:relevant}
        } \\

     \subfloat[Proxy Topic, Click Non-Relevant] {
       \centering 
\begin{tabular}{@{}L{2cm}C{1cm}C{1cm}C{1cm}C{1cm}C{1cm}C{1cm}C{1cm}C{1cm}C{1cm}C{1cm}C{1cm}|C{1cm}@{}}
\cmidrule[1pt]{1-13}

\textbf{$\delta^{*}$}   & 0.50  & 0.50    & 0.50    & 0.50  & 0.50  & 0.50  & 0.50  & 0.50  & 0.50 & 50\%  & 0.50  & 0.50  \\ \cmidrule[0.5pt]{1-13}

\textbf{SEM Range}  & & & & \multicolumn{5}{l}{\textbf{$\delta^{*}$: }  0.0 - 0.0 } \\ \cmidrule[0.5pt]{1-13}

        \end{tabular}
        \label{tbl:plausible:detect:proxy:non:relevant}
        } \\
        
     \subfloat[Proxy Topic, Click All] {
       \centering 
\begin{tabular}{@{}L{2cm}C{1cm}C{1cm}C{1cm}C{1cm}C{1cm}C{1cm}C{1cm}C{1cm}C{1cm}C{1cm}C{1cm}|C{1cm}@{}}
\cmidrule[1pt]{1-13}

\textbf{$\delta^{*}$}   & 0.50  & 0.50    & 0.50    & 0.50  & 0.50  & 0.50  & 0.50  & 0.50  & 0.50 & 50\%  & 0.50  & 0.50  \\ \cmidrule[0.5pt]{1-13}

\textbf{SEM Range}  & & & & \multicolumn{5}{l}{\textbf{$\delta^{*}$: }  0.0 - 0.0 } \\ \cmidrule[0.5pt]{1-13}

        \end{tabular}
        \label{tbl:plausible:detect:proxy:all}
        } \\
        
      \subfloat[Proxy Topic, Click 2 Random Items] {
       \centering 
\begin{tabular}{@{}L{2cm}C{1cm}C{1cm}C{1cm}C{1cm}C{1cm}C{1cm}C{1cm}C{1cm}C{1cm}C{1cm}C{1cm}|C{1cm}@{}}
\cmidrule[1pt]{1-13}
        

\textbf{$\delta^{*}$}   & 0.50  & 0.50    & 0.50    & 0.50  & 0.50  & 0.50  & 0.50  & 0.50  & 0.50 & 50\%  & 0.50  & 0.50  \\ \cmidrule[0.5pt]{1-13}

\textbf{SEM Range}  & & & & \multicolumn{5}{l}{\textbf{$\delta^{*}$: }  0.0 - 0.0 } \\ \cmidrule[0.5pt]{1-13}
        \end{tabular}
        \label{tbl:plausible:detect:proxy:random:2}
        }                              
\end{scriptsize}
}
\end{table*}%

Measurements of plausible deniability are reported in Table~\ref{tbl:proxy:clicks} for the proxy topic model and for each of the click models discussed earlier. The results indicate that $\delta^{*}=0$ for all sensitive topics and combinations of click models tested, as might be expected from the fact that the True Positive detection rate is zero. 

This result is encouraging, especially in light of the negative results in previous sections for other obfuscation approaches.  It suggests use of sequences of queries on uninteresting proxy topics may provide a defence of plausible deniability. The trade-offs for the user include the overhead of maintaining proxy topics and associated queries and the additional resources required to issue proxy topic queries in a consistent way. However since both of these tasks were readily automated during our testing it seems reasonable that these trade-offs could be readily managed by software in a way that is essentially transparent to the user.

\section{Conclusions and Discussion}
The \PRIEPS framework was used to assess threats to personal online privacy from search engine profiling and to investigate the effectiveness of a number of privacy defence tactics in protecting a user from such profiling. A comprehensive measurement program using online search engines indicated that sensitive topic learning is readily detectable with high confidence for a range of topics generally regarded as sensitive. Using a plausible deniability model of threat assessment we show that topic learning results in measurable impacts on the ability of a user to deny their interest in all sensitive topics tested. 

Our experiments indicate that revealing queries provide a significant signal for search engine adaptation. While user clicks provide additional feedback, we do not observe the same degree of associated learning with click behaviour as is observed with revealing queries. Our experiments show that injecting coherent query noise associated with ``proxy topics'' that are uninteresting to the user while capable of generating commercial content provides observable privacy protection in the case of a plausible deniability threat model.  

The approach to detection of topic profiling used in this paper does not distinguish between items based on rank or order on the page. As discussed in Section~\ref{sec:related:work}, user click patterns may be used by recommender systems to rank page content, placing content likely to attract user clicks in more prominent positions on pages. In our experiments, we observed changes in volume of advert content on samples of probe query response pages. There are several plausible avenues of investigation that may help explain the mechanism behind this, such as user click patterns and the semantics of the true and noise queries chosen. As discussed in Section~\ref{sec:related:work}, how the semantics of queries and the interaction between user click-models and content ranking may impact user privacy is beyond the scope of this current paper.

Our results suggest existing defences based on simple query injection are not adequate to prevent sensitive topic learning. We observe that noise query injection as a privacy defence strategy is most effective when queries are associated with ``proxy topics'' rather than random noise. Our results also suggest that defences based on click obfuscation approaches are also unlikely to be successful. 

These results reflect perhaps the sophistication of the learning already employed by search engines. To be robust, however, privacy defence strategies should recognise that modern recommender systems learn new topics fast, exploit multiple signals, and sustain learning  over the lifetime of query sessions. Our observation that proxy topics provide some relief indicates that defence is not impossible, but suggests that increasingly sophisticated approaches are required in the face of ever improving search engine capability. In choosing proxy topics, for example, a user must be careful to not stimulate unintended learning of the proxy topics which may influence the utility of future search results.

Overall our results point towards an ``arms race'' where search engine capability is continuously evolving. In this setting, even if injection of proxy topic sessions were to become widely deployed then we can reasonably expect search engines to respond with more sophisticated learning strategies.  Our results also point towards the fact that the text in search queries plays a key role in search engine learning. While perhaps obvious, this observation reinforces the user's need to be circumspect about the queries that they ask if they want to avoid search engine learning of their interests.


\bibliographystyle{unsrt}
\begin{small}
\bibliography{noisy}
\end{small}
\newpage
\appendix
\section{\PRIEPS Results - ``Click Relevant'' All Noise Levels}
\label{sec:backm:noise:relevant}
\begin{table*}[h]
\centering
\captionsetup{position=bottom}
\caption{Comparison of measured detection rate for various noise models (all with ``relevant  clicks'').\label{tbl:distinguish:noise:2}}
\setlength{\tabcolsep}{0.1cm}
\begin{scriptsize}
    \subfloat[No Noise] {
        \centering
\begin{tabular}{@{}L{2cm}C{1cm}C{1cm}C{1cm}C{1cm}C{1cm}C{1cm}C{1cm}C{1cm}C{1cm}C{1cm}C{1cm}|C{1cm}@{}}
\cmidrule[1pt]{1-13}

\textbf{Reference Topic} & \rot{\textbf{anorexia}} & \rot{\textbf{bankrupt}} & \rot{\textbf{diabetes}} & \rot{\textbf{disabled}} & \rot{\textbf{divorce}} & \rot{\textbf{gambling}} & \rot{\textbf{gay}} & \rot{\textbf{location}} & \rot{\textbf{payday}} & \rot{\textbf{prostate}} & \rot{\textbf{unemployed}} & \rot{\textbf{Average}}  \\ \cmidrule[1pt]{1-13} \\

\cmidrule[1pt]{1-13}
\textbf{True Detect}    & 100.0\%  & { }98.4\%  & 100.0\%  & { }94.9\%  & 100.0\%  & { }96.6\%  & 100.0\%  & 100.0\%  & 100.0\%  & 100.0\%  & 100.0\%  & 99.1\%  \\

\textbf{False Detect}   & { }0.0\%  & { }0.0\%  & { }0.0\%  & { }0.0\%  & { }0.0\%  & { }0.0\%  & { }0.0\%  & { }0.0\%  & { }0.0\%  & { }0.0\%  & { }0.0\%  & 0.0\%  \\ \cmidrule[0.5pt]{1-13}

\textbf{SEM Range}  & & \multicolumn{5}{l}{ \textbf{False Detect:}  0.0\% - 0.0\%}  & & \multicolumn{5}{l}{\textbf{True Detect:}  0.0\% - 0.2\%, } \\ \cmidrule[0.5pt]{1-13}
        \end{tabular}        
        \label{tbl:no:noise} 
        }\\
    \subfloat[Low Noise] {
        \centering
        \begin{tabular}{@{}L{2cm}C{1cm}C{1cm}C{1cm}C{1cm}C{1cm}C{1cm}C{1cm}C{1cm}C{1cm}C{1cm}C{1cm}|C{1cm}@{}}
        \cmidrule[1pt]{1-13}

\textbf{True Detect}    & 100.0\%  & 100.0\%  & 100.0\%  & { }97.4\%  & 100.0\%  & { }96.8\%  & 100.0\%  & 100.0\%  & { }91.7\%  & 100.0\%  & 100.0\%  & 98.7\%  \\

\textbf{False Detect}   & { }0.0\%  & { }0.0\%  & { }0.0\%  & { }0.0\%  & { }0.0\%  & { }0.0\%  & { }0.0\%  & { }0.0\%  & { }0.0\%  & { }0.0\%  & { }0.0\%  & 0.0\%  \\ \cmidrule[0.5pt]{1-13}

\textbf{SEM Range}  & & \multicolumn{5}{l}{ \textbf{False Detect:}  0.0\% - 0.0\%}  & & \multicolumn{5}{l}{\textbf{True Detect:}  0.0\% - 0.5\%, } \\ \cmidrule[0.5pt]{1-13}
       
        \end{tabular}        
        \label{tbl:low:noise} 
        }\\        
     \subfloat[Medium Noise] {
       \centering 
        \begin{tabular}{@{}L{2cm}C{1cm}C{1cm}C{1cm}C{1cm}C{1cm}C{1cm}C{1cm}C{1cm}C{1cm}C{1cm}C{1cm}|C{1cm}@{}}
        \cmidrule[1pt]{1-13} 
        

\textbf{True Detect}    & 100.0\%  & { }99.3\%  & 100.0\%  & { }96.8\%  & 100.0\%  & { }93.3\%  & 100.0\%  & 100.0\%  & { }98.1\%  & 100.0\%  & 100.0\%  & 98.9\%  \\

\textbf{False Detect}   & { }0.0\%  & { }0.0\%  & { }0.0\%  & { }0.0\%  & { }0.0\%  & { }0.0\%  & { }0.0\%  & { }0.0\%  & { }0.0\%  & { }0.0\%  & { }0.0\%  & 0.0\%  \\ \cmidrule[0.5pt]{1-13}

\textbf{SEM Range}  & & \multicolumn{5}{l}{ \textbf{False Detect:}  0.0\% - 0.0\%}  & & \multicolumn{5}{l}{\textbf{True Detect:}  0.0\% - 0.2\%, } \\ \cmidrule[0.5pt]{1-13}

        \end{tabular}
        \label{tbl:med:noise}
        } \\
      \subfloat[High Noise] {
       \centering 
        \begin{tabular}{@{}L{2cm}C{1cm}C{1cm}C{1cm}C{1cm}C{1cm}C{1cm}C{1cm}C{1cm}C{1cm}C{1cm}C{1cm}|C{1cm}@{}}
        \cmidrule[1pt]{1-13} 
        

\textbf{True Detect}    & 100.0\%  & 100.0\%  & 100.0\%  & { }97.5\%  & { }98.1\%  & { }96.7\%  & 100.0\%  & 100.0\%  & 100.0\%  & 100.0\%  & 100.0\%  & 99.3\%  \\

\textbf{False Detect}   & { }0.0\%  & { }0.0\%  & { }0.0\%  & { }0.0\%  & { }0.0\%  & { }0.0\%  & { }0.0\%  & { }0.0\%  & { }0.0\%  & { }0.0\%  & { }0.0\%  & 0.0\%  \\ \cmidrule[0.5pt]{1-13}

\textbf{SEM Range}  & & \multicolumn{5}{l}{ \textbf{False Detect:}  0.0\% - 0.0\%}  & & \multicolumn{5}{l}{\textbf{True Detect:}  0.0\% - 0.1\%, } \\ \cmidrule[0.5pt]{1-13}

        \end{tabular}
        \label{tbl:high:noise}
        }                              
\end{scriptsize}
\end{table*}
\newpage
\section{Proxy Topics and Queries}
\label{sec:backm:proxy:topics:queries}
The layout follows a Python configuration file format. 

Each proxy topic labels a section, for example, \emph{[tickets]}.

Keywords are included to check for adverts relevant to the proxy topic and are listed after a \emph{keywords} tag. 

Queries are listed 1 per line, after a \emph{queries} tag.
\begin{scriptsize}
\begin{multicols}{2}
\begin{verbatim}
[tickets]
keywords=concerts croke park tickets cheap deal
queries=concerts in croke park
tickets for croke park
buy tickets croke park
matches in croke park
garth brookes concert dublin
when is garth brookes coming to dublin?
electric picnic in 2016?
concerts in ireland 2016
marley park concerts?
bruce springsteen in dublin
taylor swift concert tickets?

[vacation]
keywords=holidays vacation paris 
hotels flights cheap deal
queries=travel
		flights to paris
		prices for disney paris tickets
		paris tours
		tickets for paris trains
		holiday homes in france
		what is a gite?
		holiday parks caravan france
		car hire in france
		do I need travel insurance for france?

[car]
keywords=new car deal finance garage motor 
free tax trade-in
queries=new care prices
		motor loans
		toyota dealer near here
		check resale value of my car
		trade in my old car
		auto trader used cars
		how often is nct?
		how much is my car worth for trade in?
		nissan dealer near me?
		best prices for seat cars?
\end{verbatim}
\end{multicols}
\end{scriptsize}

\end{document}